\documentclass[journal]{IEEEtran}
\IEEEoverridecommandlockouts

\usepackage{cite}
\usepackage{amsmath,amssymb,amsfonts}
\usepackage{amssymb}
\usepackage{algpseudocode}
\usepackage{amsfonts}
\usepackage{graphicx}
\usepackage{fancyhdr}
\usepackage{cases}
\usepackage{textcomp}
\usepackage{extarrows}
\usepackage{amsthm}
\usepackage{amssymb}
\usepackage{algorithm,algpseudocode}
\usepackage{orcidlink}
\hypersetup{hidelinks}

\allowdisplaybreaks

\usepackage{algorithm}
\usepackage{algpseudocode}

\usepackage{multirow}
\usepackage{caption}
\usepackage{float} 
\usepackage{subcaption}
\usepackage{makecell}
\usepackage{booktabs}

\usepackage{hyperref}
\usepackage{caption}

\usepackage{multirow,tabularx}
\usepackage{multicol,mwe,float,subcaption}
\usepackage{mathtools}
\usepackage{xcolor}
\usepackage[english]{babel}
\usepackage{amsthm}
\usepackage[numbers,sort&compress]{natbib}
\algnewcommand{\Inputs}[1]{%
  \State \textbf{Inputs:}
  \Statex \hspace*{\algorithmicindent}\parbox[t]{.8\linewidth}{\raggedright #1}
}
\algnewcommand{\Initialize}[1]{%
  \State \textbf{Initialization:}
  \Statex \hspace*{\algorithmicindent}\parbox[t]{.8\linewidth}{\raggedright #1}
}
\usepackage{caption}
\newtheorem{theorem}{Theorem}

\newtheorem{assumption}{Assumption} 

\def\BibTeX{{\rm B\kern-.05em{\sc i\kern-.025em b}\kern-.08em
    T\kern-.1667em\lower.7ex\hbox{E}\kern-.125emX}}
    
\newcommand{\Hmat}{\mathbf{H}}  
\newcommand{\Pmat}{{\mathbf{P}}}
\newcommand{\Ymat}{{\mathbf{Y}}}

\newcommand{\p}{{\mathbf{p}}}
\newcommand{\n}{{\mathbf{n}}}
\newcommand{\y}{{\mathbf{y}}}

\begin{document}

\bstctlcite{BSTcontrol}

\title{\fontsize{22 pt}{\baselineskip}\selectfont Recursive Flow: A Generative Framework for MIMO Channel Estimation}
	\author{
	\IEEEauthorblockN{
    Zehua Jiang$^{\orcidlink{0009-0003-2550-3368}}$,
	Fenghao Zhu$^{\orcidlink{0009-0006-5585-7302}}$,
	Chongwen Huang$^{\orcidlink{0000-0001-8398-8437}}$,
    Richeng Jin$^{\orcidlink{0000-0002-1480-585X}}$,
	Zhaohui Yang$^{\orcidlink{0000-0002-4475-589X}}$,\\
	Xiaoming Chen$^{\orcidlink{0000-0001-7747-6646}}$,~\IEEEmembership{Senior Member,~IEEE},
	Zhaoyang Zhang$^{\orcidlink{0000-0003-2346-6228}}$,~\IEEEmembership{Senior Member,~IEEE},
    and M\'{e}rouane~Debbah$^{\orcidlink{0000-0001-8941-8080}}$,~\IEEEmembership{Fellow,~IEEE}}

\thanks{
The work was supported National Natural Science Foundation of China under Grant 62331023 and 62394292, and China National Key R\&D Program under Grant 2025ZD1301900 and 2021YFA1000500, and Fundamental Research Funds for the Central Universities and Zhejiang University Global Partnership Fund.

Z. Jiang, F. Zhu, R. Jin, Z. Yang, X. Chen and Z. Zhang are with College of Information Science and Electronic Engineering, Zhejiang University, Hangzhou 310027, China. (E-mails: 
\href{mailto:jiangzehua@zju.edu.cn}{\{jiangzehua},
\href{mailto:zjuzfh@zju.edu.cn}{zjuzfh},
\href{mailto:richengjin@zju.edu.cn}{richengjin},
\href{mailto:yang_zhaohui@zju.edu.cn}{yang\_zhaohui},
\href{mailto:chen_xiaoming@zju.edu.cn}{chen\_xiaoming}, 
\href{mailto:ning_ming@zju.edu.cn}{ning\_ming\}@zju.edu.cn}).}

\thanks{
C. Huang is with College of Information Science and Electronic Engineering, Zhejiang University, Hangzhou 310027, China, Zhejiang Provincial Key Laboratory of Multi-Modal Communication Networks and Intelligent Information Processing, and National Key Laboratory of Millimeter-Wave and Terahertz Remote Sensing, Hangzhou 310027, China. (E-mails: 
\href{mailto:chongwenhuang@zju.edu.cn}{chongwenhuang@zju.edu.cn}).}

\thanks{M. Debbah is with  the Research Institute for Digital Future, Khalifa University, 127788 Abu Dhabi, UAE (E-mail: \href{mailto:merouane.debbah@ku.ac.ae}{merouane.debbah@ku.ac.ae}).}
}

\maketitle
\pagestyle{empty}  
\thispagestyle{empty} 

\begin{abstract}
Channel estimation is a fundamental challenge in massive multiple-input multiple-output systems, where estimation accuracy governs the spectral efficiency and link reliability. In this work, we introduce Recursive Flow (RC-Flow), a novel solver that leverages pre-trained flow matching priors to robustly recover channel state information from noisy, under-determined measurements. Different from conventional open-loop generative models, our approach establishes a closed-loop refinement framework via a serial restart mechanism and anchored trajectory rectification. By synergizing flow-consistent prior directions with data-fidelity proximal projections, the proposed RC-Flow achieves robust channel reconstruction and delivers state-of-the-art performance across diverse noise levels, particularly in noise-dominated scenarios. The framework is further augmented by an adaptive dual-scheduling strategy, offering flexible management of the trade-off between convergence speed and reconstruction accuracy. Theoretically, we analyze the Jacobian spectral radius of the recursive operator to prove its global asymptotic stability. Numerical results demonstrate that RC-Flow reduces inference latency by two orders of magnitude while achieving a 2.7 dB performance gain in low signal-to-noise ratio regimes compared to the score-based baseline.

\end{abstract}


\section{Introduction}\label{sec:intro}

The emergence of sixth-generation (6G) wireless networks promises to deliver hyper-connectivity and terabits-per-second data rates \cite{you2021towards, wang2023road}, with massive multiple-input multiple-output (MIMO) serving as a fundamental enabler to unlock spectral efficiency gains \cite{larsson2014massive, lu2014overview}. By deploying large-scale antenna arrays, future base stations can exploit high-resolution spatial multiplexing to serve massive numbers of users simultaneously \cite{tian2025analytical}. However, the theoretical benefits of such massive arrays are fundamentally predicated on the availability of precise channel state information (CSI) at the transmitter \cite{venugopal2017channel}. As the number of antennas scales to hundreds or even thousands, the conventional strategy of orthogonal pilot transmission incurs a prohibitive signaling overhead that consumes valuable coherence resources, thereby severely degrading the effective system throughput \cite{alkhateeb2014channel, schniter2014channel}. Consequently, acquiring high-fidelity CSI from a limited number of pilot measurements has evolved into a severely ill-posed linear inverse problem \cite{donoho2006compressed}, representing a critical bottleneck that necessitates algorithmic breakthroughs beyond traditional estimation paradigms.
\par

To mitigate the ill-posedness of channel estimation, conventional approaches resort to statistical or structural regularization to constrain the solution space. The linear minimum mean square error (LMMSE) estimator, widely regarded as the performance benchmark, optimally minimizes the estimation error for Gaussian channels by leveraging second-order statistics \cite{nayebi2017semi}. However, its practical deployment is often hindered by the difficulty of acquiring accurate channel covariance matrices, particularly in highly dynamic environments where channel statistics vary rapidly. Alternatively, compressed sensing (CS) techniques have emerged as a powerful paradigm \cite{donoho2006compressed}, exploiting the sparsity of millimeter-wave channels in the angular domain to recover signals from undersampled measurements . Algorithms such as approximate message passing (AMP) \cite{schniter2014channel} and atomic norm minimization \cite{bhaskar2013atomic} have demonstrated theoretical guarantees under strict sparsity conditions. Nevertheless, these model-based methods fundamentally hinge on ideal assumptions. In realistic scenarios characterized by grid mismatch, energy leakage, or rich scattering, the sparsity assumption is frequently violated, leading to significant performance degradation and reconstruction artifacts.
\par

In order to circumvent the reliance on rigid analytical priors, the research community has pivoted towards data-driven deep learning paradigms, which learn the complex inverse mapping directly from training data. Seminal works have employed convolutional neural networks to estimate channel matrices by exploiting spatial correlations \cite{he2018deep, soltani2019deep}. Furthermore, deep unfolding networks, such as Learned D-AMP \cite{metzler2017learned}, unfold iterative algorithms into trainable layers, aiming to combine the interpretability of model-based methods with the expressivity of neural networks. Despite their computational efficiency, these discriminative approaches typically minimize a pixel-wise mean squared error (MSE) loss. Mathematically, such objectives drive the network to approximate the conditional mean of the posterior distribution. While minimizing error variance, this `regression to the mean' behavior inevitably leads to over-smoothed reconstructions \cite{zhang2017beyond}, resulting in the loss of critical high-frequency spatial details and texture information that are essential for high-precision beamforming.

\par

Recently, generative artificial intelligence has empowered physical layer modeling to transcend the constraints of deterministic regression, enabling a more faithful capture of the inherent stochastic complexity of wireless channels \cite{fan2025generative, khoramnejad2025generative, shahid2025large, zhu2025wireless}. Unlike discriminative models that collapse the posterior into a single point estimate, generative approaches explicitly learn the underlying data distribution by reversing a gradual corruption process. Pioneering frameworks, such as denoising diffusion probabilistic models \cite{NEURIPS2020_4c5bcfec} and score-based generative models \cite{song2021scorebased}, estimate the score function to iteratively refine Gaussian noise into high-fidelity channel realizations. These methods have been applied in channel synthesis \cite{lee2025generating} and semantic communications \cite{wu2024cddm}, demonstrating superior capability in modeling high-frequency spatial features. Building upon these diffusion foundations, flow matching (FM) has emerged as a cutting-edge generative paradigm \cite{lipman2023flow, papamakarios2021normalizing, kobyzev2020normalizing, liu2023flow}. By regressing a time-dependent vector field approximating the optimal transport path between prior and data distributions \cite{montesuma2024recent}, FM rectifies generative trajectories, enabling faster and more deterministic sampling than stochastic diffusion. In the broader context for linear inverse problems, general-purpose solvers such as FLOWER \cite{pourya2025flower} and PnP-Flow \cite{martin2025pnpflow} have been proposed. Treating the pre-trained flow as a generative prior, these algorithms alternate flow-based denoising and measurement-consistency projections to steer trajectories toward the data manifold and measurement subspace intersection.
\vspace{0.5em}

\par

Despite the theoretical potential of generative models, their practical deployment for high-dimensional MIMO channel estimation is hindered by specific implementation constraints. A prominent branch of recent flow-based estimators focuses on architectural innovations to accelerate inference. Specifically, the approach in \cite{liu2025flow} applies standard flow matching to model the conditional channel distribution, while \cite{fesl2024diffusion} introduces a signal-to-noise ratio (SNR) aware truncation strategy for the reverse diffusion process to reduce computational complexity. Furthermore, \cite{jiang2025one} leverages the average velocity field to collapse the integration trajectory, enabling channel estimation in a single step. While algorithmically distinct, these methods typically rely on the conjugate transpose of pilot matrices, implicitly assuming pilot orthogonality. In under-determined massive MIMO systems where pilots are limited, this assumption fails, causing severe aliasing artifacts.

\par

Alternatively, the foundational score-based method \cite{arvinte2022mimo} provides a theoretically rigorous solution for such under-determined systems by employing annealed Langevin dynamics to sample from the posterior. However, the stochastic nature of the dynamics necessitates the continuous injection of noise to correct trajectory errors, requiring thousands of iterative steps that result in prohibitive computational latency. Aiming to reduce this latency, subsequent accelerated deterministic solvers have been proposed. For instance, \cite{zhou2025generative} derives a closed-form approximation of the likelihood score to guide the diffusion reverse process, while \cite{chen2025generative, diao2025robust} utilize Tweedie’s formula \cite{efron2011tweedie, kim2021noise2score} or energy-based variational inference to approximate the posterior mean in fewer steps. Nevertheless, critical limitations persist across these deterministic modifications. Their data-consistency updates typically rely on explicit gradient descent steps, where determining the optimal step size presents a significant challenge. Besides, despite the reduction in sampling steps compared to score-based methods \cite{arvinte2022mimo}, these solvers still require hundreds of iterations, which remains computationally prohibitive for low-latency systems. More fundamentally, they operate in an open-loop manner, where the generative trajectory is initialized once from noise and refined sequentially. In low SNR regimes, the guidance provided by the likelihood term is dominated by heavy measurement noise rather than the true signal residual, causing the solver to drift towards a biased solution.

\par

To resolve the systematic bias inherent in open-loop solvers and the instability of gradient-based physical updates, a paradigm shift is required from `trajectory sampling' to `equilibrium finding'. Drawing inspiration from deep equilibrium models \cite{bai2019deep}, which define the output of a network as the fixed point of a non-linear transformation, we propose to reformulate the generative channel estimation process as a closed-loop iteration. The main contributions of this paper can be summarized as follows:

\begin{itemize}
\item We introduce Recursive Flow (RC-Flow), a novel generative framework for high-dimensional MIMO channel estimation. By reformulating the inference process as a closed-loop fixed-point iteration and leveraging anchored trajectory rectification, we theoretically guarantee global asymptotic stability towards a fixed point, providing a robust mathematical foundation for generative channel recovery.
\item We propose a closed-form proximal projection mechanism to replace heuristic gradient-based updates. This approach ensures unconditional data consistency throughout the iterative process and eliminates the need for tedious step-size tuning, thereby enhancing the solver's robustness against measurement noise compared to standard diffusion baselines.
\item We design an adaptive dual-scheduling strategy ($ \lambda,\beta$) that dynamically reallocates the computational budget during inference. This mechanism allows the solver to adaptively trade off between convergence speed and restoration precision, enabling the system to meet stringent low-latency requirements essential for real-time 6G applications.
\item Extensive simulations demonstrate that RC-Flow achieves SOTA performance across diverse scenarios. In low-SNR regimes, it effectively suppresses noise-induced hallucinations, significantly outperforming existing baselines. In high-SNR regimes, it matches the accuracy of computationally expensive methods while reducing inference latency by orders of magnitude, validating its practical superiority. The code for this paper is available online in \cite{sourcecode}.
\end{itemize}
\par
The paper is organized as follows: Section \ref{sec:sys} presents an introduction to the system model and flow matching. Section \ref{sec:RC-Flow} introduces the proposed RC-Flow framework, followed by the presentation of simulation results in Section \ref{sec:simulation}. Finally, Section \ref{sec:conclusion} concludes the paper.
\par
\textit{Notation}: Fonts $a$, $\mathbf{a}$, and $\mathbf{A}$ denote scalars, vectors, and matrices, respectively. $\|\mathbf{a}\|_2$ represent the 2-norm of $\mathbf{a}$. $\mathbf{A}^T$, 
$\mathbf{A}^H$, $\mathbf{A}^{-1}$, $\|\mathbf{A}\|_2$ and $\|\mathbf{A}\|_F$ represent the transpose, conjugate transpose, inverse, 2-norm and Frobenius norm of $\mathbf{A}$, respectively. The $(m,n)$-th entry of $\mathbf{A}$ is denoted by $a_{mn}$, and $|\cdot|$ denotes the modulus. $\mathcal{O}$ represents the asymptotic time complexity. Finally, we can represent the diagonal matrix and trace of matrix $\mathbf{A}$ using the notations $\mathrm{diag}(\mathbf{A})$ and $\mathrm{Tr}(\mathbf{A})$, respectively.
\par


\vspace{0mm}

\section{Preliminaries}\label{sec:sys}
\subsection{Wireless System Model}
\label{subsec:system_model}
We consider a narrowband, point-to-point MIMO communication system consisting of $N_t$ transmit antennas and $N_r$ receive antennas. The wireless propagation environment is characterized by the complex-valued CSI matrix $\mathbf{H} \in \mathbb{C}^{N_r \times N_t}$. To estimate this channel, the transmitter sends a sequence of $N_p$ pilot symbols. Let $\p_i \in \mathbb{C}^{N_t}$ denote the $i$-th pilot symbol and $\mathbf{W} \in \mathbb{C}^{N_r \times N_r}$ be the receive beamforming matrix. The received signal vector for the $i$-th pilot $\y_i$ is given by:
\begin{equation}
\mathbf{y}_i = \mathbf{W}^H (\mathbf{H} \mathbf{p}_i s + \mathbf{n}_i),
\end{equation}
where $\mathbf{n}_i$ is the complex additive white Gaussian noise and $s$ is a complex-valued scalar. For the remainder of the paper, we assume $s=1$ and a fully digital receiver where $\mathbf{W} = \mathbf{I}$. The pilot $\p_i$ is constrained to have unit-amplitude entries with low-resolution phase, where each element is a randomly selected quadrature phase shift keying symbol that remains fixed for all test samples. Assuming the channel remains constant over the duration of the pilot transmission, the received signal matrix $\mathbf{Y}$ is modeled as:
\begin{equation}
\mathbf{Y} = \mathbf{H}\mathbf{P} + \mathbf{N},
\end{equation}
where $\Ymat=[\y_1, \y_2,...\y_{N_p}] \in \mathbb{C}^{N_r \times N_p}$ is the received signal matrix, $\Pmat=[\p_1, \p_2,...\p_{N_p}] \in \mathbb{C}^{N_t \times N_p}$  is the pilot matrix, $\mathbf{N}=[\n_1, \n_2,...\n_{N_p}] \in \mathbb{C}^{N_r \times N_p}$ is the additive noise matrix. The entries of $\mathbf{N}$ are typically assumed to be independent and identically distributed complex white Gaussian noise with zero mean and variance $\sigma_{\text{pilot}}^2$.
\par
The objective of channel estimation is to recover the high-dimensional matrix $\Hmat$ from the observed measurements $\Ymat$, given the knowledge of the pilot matrix $\Pmat$. We define the pilot density as $\alpha = N_p/N_t$. In practical mmWave scenarios, pilot resources are often severely constrained ($\alpha<1$), thereby rendering the recovery of an inherently under-determined inverse problem.

\subsection{Continuous Normalizing Flows}
Continuous normalizing flows represent a class of generative models that transform a simple prior distribution $p_0(\mathbf{H})$ (e.g., a standard Gaussian distribution) into a complex data distribution $p_1(\mathbf{H})$ through a continuous-time generative process \cite{kobyzev2020normalizing, papamakarios2021normalizing}.
This process is characterized by a flow, denoted as:
\begin{equation}
    \phi_t:[0,1]\times \mathbb{R}^{d} \xrightarrow{} \mathbb{R}^{d},
\end{equation}
where $d$ represents the dimension of the data, the symbol $\times$ and $\xrightarrow{}$ denote the Cartesian Product and the mapping between data domains, respectively. Intuitively, $\phi_t$ moves a sample from its initial state at $t=0$ to a target state at $t=1$. Mathematically, for any time $t$ and any channel matrix $\mathbf{H}$ in the signal space, the flow is defined by an ordinary differential equation (ODE):
\begin{equation}
\frac{d}{dt} \phi_t(\mathbf{H}) = \mathbf{V}_t(\phi_t(\mathbf{H})), \quad \phi_0(\mathbf{H}) = \mathbf{H}_0,
\end{equation}
where $\phi_0(\mathbf{H}) = \mathbf{H}_0$ defines the boundary condition, $\mathbf{V}_t$ is the time-dependent velocity field that determines the instantaneous direction and speed of the evolution. This sample-wise transport induces a time-dependent probability density path $p_t$. The coupling between the density $p_t$ and the velocity field $\mathbf{V}_t$ is formally described by the continuity equation:
\begin{equation}
\frac{\partial p_t(\mathbf{H})}{\partial t} + \text{div}(p_t(\mathbf{H}) \mathbf{V}_t(\mathbf{H})) = 0,
\end{equation}
which ensures that the velocity field correctly describes the evolution of the probability density over time.

\subsection{Flow Matching}
The objective of flow matching (FM) is to train a deep neural network, parameterized by $\theta$, to approximate the optimal marginal velocity field $\mathbf{V}_t$ \cite{liu2025flow, lipman2023flow, martin2025pnpflow}. The network parameterizes a time-dependent vector field $\mathbf{V}_t^{\theta}$, which is optimized to minimize the MSE-based FM loss:
\begin{equation}
\mathcal{L}_{\text{FM}}(\theta) = \mathbb{E}_{t \sim \mathcal{U}[0,1], \mathbf{H} \sim p_t(\mathbf{H})} \left[ \| \mathbf{V}_t^{\theta}(\mathbf{H}) - \mathbf{V}_t(\mathbf{H}) \|_F^2 \right].
\end{equation}
However, directly minimizing $\mathcal{L}_{\text{FM}}$ is typically infeasible in practice as both the marginal density $p_t(\mathbf{H})$ and the target field $\mathbf{V}_t$ are generally inaccessible. To circumvent this, conditional flow matching (CFM) is utilized \cite{tamir2024conditional}, which simplifies the optimization by conditioning on individual data samples $\mathbf{H}_1$ drawn from the empirical distribution $p_1(\mathbf{H})$. By defining a conditional probability path $p_t(\mathbf{H}|\mathbf{H}_1)$ and a corresponding conditional velocity field $\mathbf{V}_t(\mathbf{H}|\mathbf{H}_1)$ that transports probability mass toward $\mathbf{H}_1$, the CFM objective is formulated as:
\begin{equation}
\mathcal{L}_{\text{CFM}}(\theta) = \mathbb{E}_{t, \mathbf{H}_1 \sim p_1, \mathbf{H} \sim p_t(\mathbf{H}|\mathbf{H}_1)} \left[ \|\mathbf{V}_t^{\theta}(\mathbf{H}) - \mathbf{V}_t(\mathbf{H}|\mathbf{H}_1)\|_F^2 \right].
\end{equation}
Notably, the gradients of $\mathcal{L}_{\text{CFM}}$ and $\mathcal{L}_{\text{FM}}$ with respect to $\theta$ are equivalent, provided that the marginal velocity field is the aggregation of its conditional counterparts:
\begin{equation}
    \mathbf{V}_t(\mathbf{H}) = \int \mathbf{V}_t(\mathbf{H}|\mathbf{H}_1) \frac{p_t(\mathbf{H}|\mathbf{H}_1)p_1(\mathbf{H}_1)}{p_t(\mathbf{H})} d\mathbf{H}_1.
\end{equation}
This framework enables `simulation-free' training. Specifically, for Gaussian conditional paths, samples at any time $t$ can be generated via a closed-form reparameterization, eliminating the need for expensive ODE solvers during the training phase.
\par
In our implementation, we adopt a Gaussian conditional probability path $p_t(\mathbf{H} |\mathbf{H}_1) = \mathcal{N}(\mathbf{H} ; \mu_t(\mathbf{H}_1), \sigma_t^2 \mathbf{I})$ and employ a linear interpolation strategy to construct the flow. This trajectory is defined by the displacement map:
\begin{equation}
\psi_t(\mathbf{H}_0, \mathbf{H}_1) = t \mathbf{H}_1 + (1-t) \mathbf{H}_0,
\end{equation}
where $\mathbf{H}_0 \sim \mathcal{N}(\mathbf{0}, \mathbf{I})$ denotes the initial noise distribution, and the mean and variance evolve as $\mu_t(\mathbf{H}_1) = t \mathbf{H}_1$ and $\sigma_t^2 = (1-t)^2$, respectively. Such a construction induces a time-independent conditional velocity field:
\begin{equation}
\mathbf{V}_t(\mathbf{H}|\mathbf{H}_1) = \frac{d}{dt}\psi_t(\mathbf{H}_0, \mathbf{H}_1) = \mathbf{H}_1 - \mathbf{H}_0.
\end{equation}
Compared to curved diffusion-based paths, this rectified linear trajectory significantly accelerates training convergence and enhances sampling efficiency during inference. Finally, to align with the estimation task, we define the generative process via a reverse-time flow, evolving from $t=1$ back to $t=0$ (ground-truth). 

\vspace{0mm}

\begin{figure*}[t!]\vspace{0mm}
	\begin{center}
		\centerline{\includegraphics[width=0.9\textwidth]{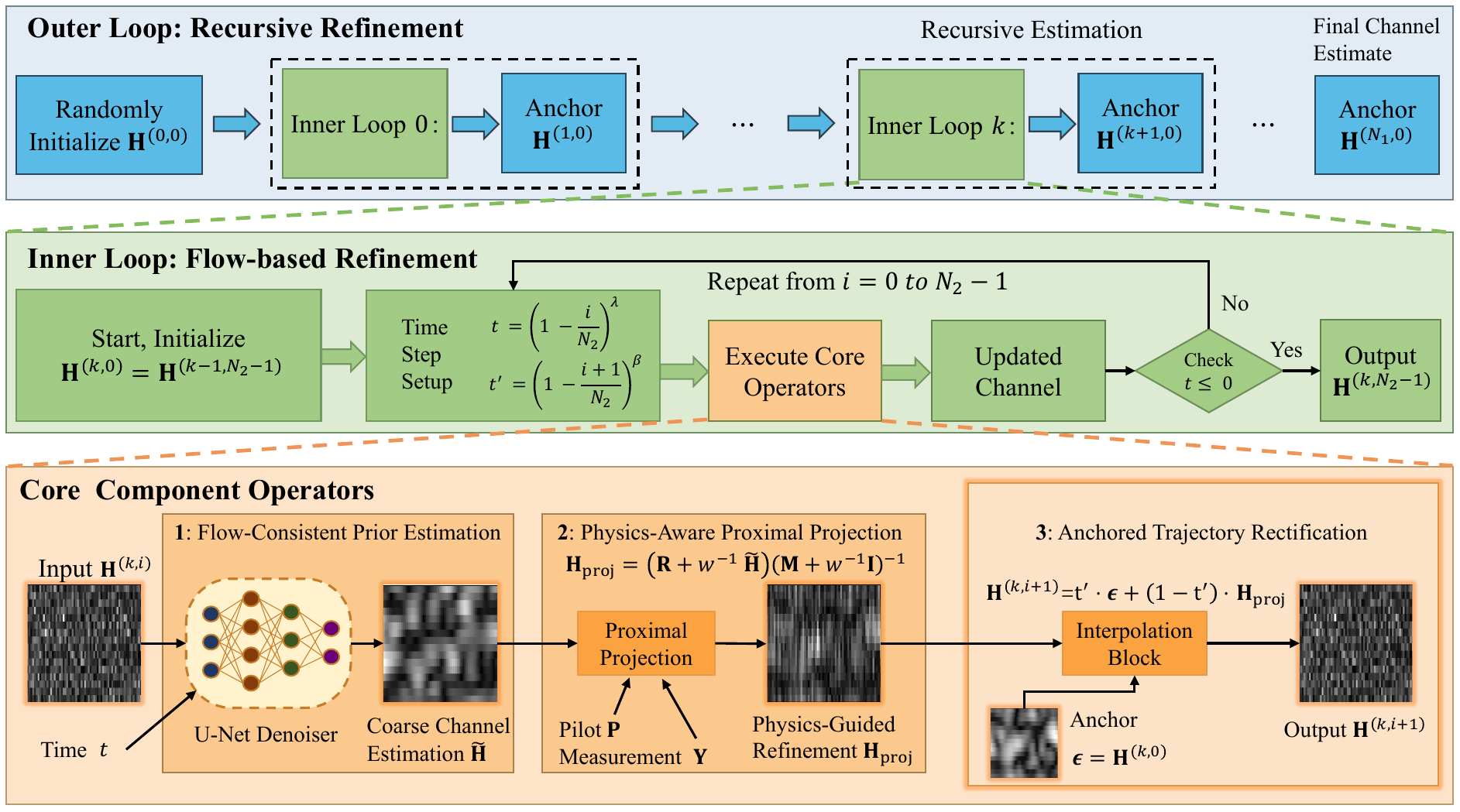}}  \vspace{-0mm}
		\captionsetup{font=footnotesize, name={Fig.}, labelsep=period} 
		\caption{\, Schematic of the RC-Flow algorithm. The method employs a nested loop structure: the Outer Loop recursively resets the anchor to correct trajectory drift, while the Inner Loop refines the estimate. The Core Component Operators  combine a deep flow prior, physics-aware proximal projection, and anchored rectification to achieve robust reconstruction.}
		\label{fig:Algorithm_Chart} \vspace{-8mm}
	\end{center}
\end{figure*}

\section{Recursive Flow Framework}\label{sec:RC-Flow}
In this section, we present the RC-Flow framework, as illustrated in Fig. \ref{fig:Algorithm_Chart} and summarized in Algorithm \ref{alg:Inference_Algorithm}. We structure the presentation hierarchically: first, we define the core component operators that serve as the fundamental building blocks; second, we detail the inner loop mechanism driven by an adaptive dual-scheduling strategy; third, we introduce the outer loop recursive anchor refinement for trajectory rectification; finally, we provide a rigorous theoretical analysis establishing the existence and local asymptotic stability of the solver's fixed point.

\subsection{Core Component Operators}
\paragraph{Flow-Consistent Prior Estimation}

To establish an informative generative prior for channel estimation, we utilize the CFM framework pre-trained on empirical channel realizations. In the training phase, the network, parameterized by $\theta$, learns to approximate a time-dependent velocity field. For each training instance, an SNR is sampled from a predefined discrete set $\mathcal{S}_{dB}$. Based on the selected SNR, complex Gaussian noise is generated and superimposed on the ground-truth channel $\mathbf{H}$. Here, $\mathbf{H}$ is globally normalized such that $\mathbb{E}[|h_{ij}|^2]=1$. We define the probability density path through linear interpolation between the noise-free source $\mathbf{H}_0 = \mathbf{H}$ and the noisy target $\mathbf{H}_1 = \mathbf{H} + \mathbf{N}$, expressed as $\mathbf{H}_t = (1-t)\mathbf{H}_0 + t\mathbf{H}_1$. To optimize the learning efficiency across the flow trajectory, the temporal variable $t \in [0,1]$ follows a logit-normal schedule. The parameters $\theta$ are optimized by minimizing the MSE between the predicted and target velocity fields:
\begin{equation}
    \mathcal{L}(\theta) = \mathbb{E}_{t, \mathbf{H}, \mathbf{N}} \left[ \| \mathbf{V}^{\theta}_t(\mathbf{H}_t, t) - (\mathbf{H}_1 - \mathbf{H}_0) \|_F^2 \right].
    \label{eq:fm_loss}
\end{equation}
Unlike conditional generative models that tie the network to specific sensing patterns, our framework adopts a task-agnostic prior. This design decouples the underlying channel manifold from the measurement operator (e.g., pilot density $\alpha$ and matrix $\mathbf{P}$), enabling zero-shot generalization to diverse sensing configurations while ensuring strict physical data-consistency through explicit proximal projections.
\par
During inference, the objective is to recover the clean channel from a noisy intermediate state by leveraging the learned flow dynamics. Specifically, given a current state $\mathbf{H}^{(k,i)}$ and a time step $t$, the pre-trained network predicts the velocity $\mathbf{V} = \text{Model}(\mathbf{H}^{(k,i)}, t)$ tangent to the probability flow. A flow-consistent prior estimate $\tilde{\mathbf{H}}$ is then obtained by projecting the state along the trajectory back to $t=0$. This yields the following deterministic single-step estimation formula:
\begin{equation}
    \tilde{\mathbf{H}} = \mathbf{H}^{(k,i)} - t \cdot \mathbf{V}.
    \label{eq:prior_est}
\end{equation}
This operation effectively utilizes the learned generative prior to denoise the channel from the current observation.

\paragraph{Physics-Aware Proximal Projection}

While the prior estimate $\tilde{\mathbf{H}}$ is consistent with the learned distribution, it may not strictly adhere to the physical constraints imposed by the pilot observations. To enforce data consistency, we employ a proximal operator that reconciles the generative prior with the measurement likelihood. This is formulated as a regularized least-squares problem:
\begin{equation}
    \mathbf{H}_{\text{proj}} = \arg\min_{\mathbf{H}} \frac{1}{2\sigma^2_{\text{pilot}}} \|\mathbf{Y} - \mathbf{HP}\|_F^2 + \frac{1}{2w} \|\mathbf{H} - \tilde{\mathbf{H}}\|_F^2,
    \label{eq:proximal_obj}
\end{equation}
where $w$ is a time-dependent regularization parameter. We adopt a variance annealing schedule defined as $w = t^2 / (t^2 + (1-t)^2)$, which dynamically balances the data fidelity of the pilot measurements with the flow-consistent prior.
\par
As a quadratic optimization problem, \eqref{eq:proximal_obj} admits a unique closed-form solution:
\begin{equation}
    \mathbf{H}_{\text{proj}} = (\mathbf{R} + w^{-1}\tilde{\mathbf{H}}) (\mathbf{M} + w^{-1}\mathbf{I})^{-1},
    \label{eq:proximal_solution}
\end{equation}
where $\mathbf{M} = \sigma_{\text{pilot}}^{-2}\mathbf{P}\mathbf{P}^H$ and $\mathbf{R} = \sigma_{\text{pilot}}^{-2}\mathbf{Y}\mathbf{P}^H$. To circumvent the prohibitive computational complexity associated with direct matrix inversion, we leverage eigenvalue decomposition (EVD). Given that $\mathbf{M}$ is a Hermitian matrix, its EVD is expressed as:
\begin{equation}
    \mathbf{M} = \mathbf{U} \boldsymbol{\Lambda} \mathbf{U}^H,
\end{equation}
where $\mathbf{U}$ is a unitary matrix and $\boldsymbol{\Lambda}$ is a diagonal matrix of eigenvalues. Consequently, the regularized inversion term in \eqref{eq:proximal_solution} can be efficiently reformulated as:
\begin{equation}
    (\mathbf{M} + w^{-1}\mathbf{I})^{-1} = \mathbf{U} (\boldsymbol{\Lambda} + w^{-1}\mathbf{I})^{-1} \mathbf{U}^H.
\end{equation}
Since $(\boldsymbol{\Lambda} + w^{-1}\mathbf{I})$ remains diagonal, the inversion is reduced to an element-wise reciprocal of its diagonal entries, significantly reducing the per-iteration computational overhead.

\paragraph{Anchored Trajectory Rectification}

Although the proximal projection enforces data consistency, directly adopting the refined estimate $\mathbf{H}_{\text{proj}}$ may cause the state to deviate from the learned probability flow trajectory. To mitigate this trajectory drift and maintain structural integrity, we implement an anchored rectification strategy designed to preserve the generative path. Specifically, we define the latent anchor $\boldsymbol{\epsilon} = \mathbf{H}^{(k,0)}$ to represent the source of the flow. The updated state for the subsequent time step $t'$ is then reconstructed by re-interpolating between the anchor and the physically refined estimate $\mathbf{H}_{\text{proj}}$, expressed as:
\begin{equation}
    \mathbf{H}^{(k,i+1)} = t' \cdot \boldsymbol{\epsilon} + (1 - t') \cdot \mathbf{H}_{\text{proj}}.
    \label{eq:rectification}
\end{equation}
This formulation ensures that the transition from noise to data adheres to the optimal straight-line path, while effectively steering the generative evolution toward the physically consistent manifold.

\begin{algorithm}[t]
\caption{Recursive Flow Algorithm} 
\label{alg:Inference_Algorithm}
\begin{algorithmic}[1]
    \Require Observations $\mathbf{Y}$, Pilots $\mathbf{P}$, $\sigma_{\text{pilot}}$, $N_1, N_2$, $\lambda, \beta$.
    \Ensure Estimated Channel $\mathbf{H}_{\text{est}}$. 
    
    \State Initialize $\mathbf{M} = \sigma_{\text{pilot}}^{-2}\mathbf{P}\mathbf{P}^H$ and $\mathbf{R} = \sigma^{-2}_{\text{pilot}}\mathbf{Y}\mathbf{P}^H$.
    \State EVD: $\mathbf{M} = \mathbf{U} \boldsymbol{\Lambda} \mathbf{U}^H$.
    \State Initialize $\mathbf{H}^{(0,0)} \sim \mathcal{CN}(\mathbf{0}, \mathbf{I})$.
    
    \For{$k \gets 0, \dots, N_1 - 1$}
      
        \For{$i \gets 0, \dots, N_2-1$}
            \State $t = (1 - i/N_2)^{\lambda}$
            \State $t' = (1 - (i+1)/N_2)^{\beta}$
            \State $\mathbf{V} = \mathrm{Model}(\mathbf{H}^{(k, i)}, t)$ 
            \State $\tilde{\mathbf{H}} = \mathbf{H}^{(k, i)} - t \cdot \mathbf{V}$ \Comment{Denoising}
            \State $w = t^2 / (t^2 + (1 - t)^2)$
            \State $\mathbf{H}_{\text{proj}} = (\mathbf{R} + w^{-1}\tilde{\mathbf{H}})\mathbf{U}(\boldsymbol{\Lambda} +w^{-1}\mathbf{I})^{-1} \mathbf{U}^H$
            \State \Comment{Projection}
            \If{$k=0$} 
                \State $\boldsymbol{\epsilon} \sim \mathcal{CN}(\mathbf{0}, \mathbf{I})$ \Comment{Initialize Anchor}
            \Else
                \State $\boldsymbol{\epsilon} = \mathbf{H}^{(k - 1, N_2 - 1)}$ \Comment{Reset Anchor}
            \EndIf
            \State $\mathbf{H}^{(k, i+1)} = t' \cdot \boldsymbol{\epsilon} + (1 - t') \cdot \mathbf{H}_{\text{proj}}$ \Comment{Interpolation}
        \EndFor
        \State $\mathbf{H}^{(k + 1, 0)} = \mathbf{H}^{(k, N_2-1)}$
    \EndFor 
    \State $\mathbf{H}_{\text{est}}$ = $\mathbf{H}^{(N_1 - 1, N_2 - 1)}$
    \State \Return $\mathbf{H}_{\text{est}}$
\end{algorithmic}
\end{algorithm}

\subsection{Inner Loop: Flow-based Refinement}
The aforementioned operators form the core processing unit of the proposed framework, integrated into an inner loop designed to progressively refine the channel estimate from a noisy observation toward a ground-truth realization. The flow-based refinement process is structured as follows:
\par
In the $k$-th outer iteration, the inner loop is initialized with the anchor $\mathbf{H}^{(k,0)} = \mathbf{H}^{(k-1,N_2-1)}$ as dictated by the recursive anchor logic. The loop then executes $N_2$ steps ($i \in [0,N_2-1]$) to traverse the probability flow trajectory.
\par
To balance trajectory stability with the capture of rapid velocity dynamics near the data manifold, we employ an adaptive dual-scheduling strategy. This approach decouples the time coordinate for feature extraction from the temporal step used for noise removal. Specifically, for the $i$-th inner step, the current time $t$ and target time $t'$ are governed by two distinct polynomial schedules:
\begin{align}
    t  &= (1 - i/N_2)^\lambda, \label{eq:lambda_sched} \\
    t' &= (1 - (i+1)/N_2)^\beta. \label{eq:beta_sched}
\end{align}
Where $\lambda$ modulates the time embedding of the network to govern the granularity of prior estimation, while $\beta$ regulates the rectification step size and the rate of anchor removal. This mechanism facilitates rapid progression in high-noise regimes while ensuring fine-grained adjustments as the channel structure emerges. The impact of $\lambda$ and $\beta$ configurations on convergence speed and estimation performance is further analyzed in Section \ref{sec:simulation}.
\par
Given the time indices $t$ and $t'$, the loop updates the channel state $\mathbf{H}^{(k,i)}$ by sequentially applying the operators defined in \eqref{eq:prior_est}--\eqref{eq:rectification}. This sequence integrates flow-consistent prior extraction, physics-aware projection, and trajectory rectification to generate the subsequent state $\mathbf{H}^{(k,i+1)}$. This procedure repeats until $i = N_2-1$ (where $t'=0$). The final state $\mathbf{H}^{(k, N_2-1)}$ serves as both the refined channel estimate for the current outer iteration and the anchor for the recursive update in the subsequent stage.

\subsection{Outer Loop: Recursive Anchor Refinement}
While the inner loop integrates the core operators, a singular traversal of the probability flow trajectory often fails to obtain an optimal channel estimate. The substantial gap between the Gaussian prior and the observed posterior, coupled with nonlinear corrections from proximal projections, typically results in suboptimal convergence. To address this, we propose a recursive anchor refinement strategy. In contrast to standard single-pass inference, this approach leverages the output of the current stage to re-initialize the anchor for the subsequent stage, effectively moving the generative flow source closer to the target data manifold.
\par
Formally, the anchor serves as the source state of the probability path at $t=1$, functioning as the fixed reference point for the trajectory rectification step in \eqref{eq:rectification}. In this recursive scheme, the final estimate from the $k$-th inner loop, denoted as $\mathbf{H}^{(k, N_2-1)}$, is assigned as the anchor for the subsequent iteration. This update rule is formulated as:
\begin{align}
    &\Hmat^{(k+1,0)} = \mathbf{H}^{(k, N_2-1)}, \\
    &\Hmat^{(0,0)} \sim \mathcal{CN}(\mathbf{0}, \mathbf{I}).
\end{align}
where $\mathbf{H}^{(0,0)}$ denotes the initial Gaussian noise. This recursion proceeds for $N_1$ outer iterations, after which the final state $\mathbf{H}^{(N_1,0)}$ is output as the estimated channel matrix.
\par
This anchor-based refinement functions as a critical self-correction mechanism. By re-initializing the flow source with the most recent estimate, the algorithm effectively compensates for trajectory drift induced by the finite discretization steps of the inner loop. Consequently, this recursive process progressively steers the estimate toward a high-fidelity solution that is consistent with both learned priors and physical measurements.

\subsection{Theoretical Analysis: Global Asymptotic Stability}
\label{subsec:theory}
We model the algorithm as a discrete dynamical system governed by the composite operator $\mathcal{T}: \mathbb{C}^{N_r \times N_t} \to \mathbb{C}^{N_r \times N_t}$. Specifically, a single inner-loop iteration can be expressed as:
\begin{equation}
    \mathbf{H}^{(k+1, 0)} = \mathcal{T}(\mathbf{H}^{(k, 0)}) \triangleq \mathcal{A} \circ \mathcal{P} \circ \mathcal{D} (\mathbf{H}^{(k, 0)}),
\end{equation}
where $\mathcal{D}$ represents the flow-matching denoiser, $\mathcal{P}$ denotes the physics-aware proximal projection, and $\mathcal{A}$ is the anchored rectification step.
\par
To prove the existence of a fixed point $\mathbf{H}^\star = \mathcal{T}(\mathbf{H}^\star)$, we adopt the standard assumption of \textit{Bounded Denoisers}, which is widely used in the analysis of Plug-and-Play algorithms \cite{chan2016plug}.

\begin{assumption}[Bounded Denoiser Hypothesis]
    \label{assump:bounded}
    The pre-trained flow matching model acts as a bounded denoiser. This indicates that for any input state $\mathbf{H} \in \mathbb{C}^{N_r \times N_t}$, the output magnitude is bounded by a constant $B_{\text{flow}} < \infty$:
    \begin{equation}
        \| \mathcal{D}(\mathbf{H})\|_F \le B_{\text{flow}}.
    \end{equation}
\end{assumption}

Under this mild assumption, we can utilize the properties of the physical projection operator to guarantee the existence of equilibrium states.

\begin{theorem}[Existence of Fixed Point]
    \label{thm:existence}
    Consider the composite operator $\mathcal{T}$ with anchor coefficient $t'_i \in (0, 1)$ and regularization weight $w > 0$. Under Assumption \ref{assump:bounded}, there exists a compact convex set $\mathcal{K} \subset \mathbb{C}^{N_r \times N_t}$ such that $\mathcal{T}$ maps $\mathcal{K}$ to itself ($\mathcal{T}(\mathcal{K}) \subseteq \mathcal{K}$). Consequently, by Brouwer's Fixed-Point Theorem \cite{kellogg1976constructive}, there is at least one fixed point $\mathbf{H}^\star \in \mathcal{K}$.
\end{theorem}

\begin{proof}
    See Appendix \ref{app:proof_existence}.
\end{proof}


Having established the existence of a fixed point, we now analyze the stability of the RC-Flow solver. We define the total Jacobian of the outer-loop operator $\mathcal{T}$ at the fixed point $\mathbf{H}^\star$ as $\mathbf{J}_{\mathcal{T}} \triangleq \frac{\partial \mathrm{vec}(\mathbf{H}^{(k+1)})}{\partial \mathrm{vec}(\mathbf{H}^{(k)})} |_{\mathbf{H}^\star}$. In high-dimensional inverse problems, generative priors often exhibit local expansion (i.e., Lipschitz constant $> 1$) to recover high-frequency textures. However, our framework relies on the geometric interplay between the prior and the physical constraints. We encapsulate this via the following hypothesis.

\begin{assumption}[Spectral Contraction of Composite Operator]
    \label{assump:contraction}
    Let $\mathbf{J}_{\mathcal{D}, i}$ and $\mathbf{J}_{\mathcal{P}, i}$ be the Jacobians of the flow denoiser and the physics-aware projection at the $i$-th inner step, respectively. We assume that the projection operator imposes sufficient structural constraints to suppress the expansive modes of the denoiser, such that the spectral radius of the composite Jacobian $\mathbf{T}_i \triangleq \mathbf{J}_{\mathcal{P}, i}\mathbf{J}_{\mathcal{D}, i}$ satisfies:
    \begin{equation}
        \rho(\mathbf{T}_i) \le \gamma < 1.
    \end{equation}
\end{assumption}
Under this assumption, we have the following theorem:
\begin{theorem}[Global Asymptotic Stability in Generalized Metric Space]
    \label{thm:stability}
    Suppose the inner-loop operators satisfy Assumption \ref{assump:contraction}, and the anchor coefficients $\{t'_i\}_{i=0}^{N_2-1}$ are time-varying within $[0, 1)$. Then there exists a generalized vector norm $\|\cdot\|_*$ under which the total Jacobian $\mathbf{J}_{\mathcal{T}}$ is strictly contractive, i.e., $\|\mathbf{J}_{\mathcal{T}}\|_* < 1$. Consequently, the recursive iteration converges linearly to the fixed point $\mathbf{H}^\star$ for any initialization within the basin of attraction.
\end{theorem}


\begin{proof}
    See Appendix \ref{app:proof_stability}.
\end{proof}

\par
Therefore, due to the equivalence of norms in finite-dimensional spaces, convergence in $\|\cdot\|_*$ implies global asymptotic stability in the standard Euclidean norm $\|\cdot\|_2$.

\vspace{0mm}
\section{Numerical Results}\label{sec:simulation}

In this section, we provide a comprehensive numerical evaluation of the proposed RC-Flow framework across a wide range of operational conditions. Following the algorithmic established in Section \ref{sec:RC-Flow}, the performance and robustness of the scheme are systematically evaluated in the following subsections.

\begin{table}[t]
    \centering
    \caption{System Configurations}
    \label{tab:hyperparams}
    \renewcommand{\arraystretch}{1.2}
    \begin{tabular}{l|c}
        \hline
        \textbf{Parameter} & \textbf{Value} \\
        \hline
        \multicolumn{2}{c}{\textit{Training Configurations}} \\
        \hline
        Optimizer & Adam \\
        Learning Rate & $1 \times 10^{-4}$ \\
        Batch Size & 256 \\
        Total Epochs & 1600 \\
        Training SNR Range & $\{-10, -5, \dots, 30\}$\\
        Time Sample & $\sigma(\xi)$ \\
        EMA Decay Rate & 0.999 \\
        Precision & bf16 \\
        \hline
        \multicolumn{2}{c}{\textit{Inference Configurations}} \\
        \hline
        Time Schedule Parameter ($\lambda,\beta$)  & 2,2 \\
        Model Selection (Epoch) & 1600 \\
        Inner Loop Number &  \eqref{eq:Inner Loop Steps} \\
        $N_{\max}$, $N_{\min}$ &  50,3 \\
        \hline
        \multicolumn{2}{c}{\textit{Network Architecture}} \\
        \hline
        Input/Output Channels & 2 (Real/Imag) \\
        Base Model Channels & 32 \\
        Channel Multipliers & (1, 2, 4) \\
        Residual Blocks per Level & 2 \\
        Time Embedding Dimension & 128 \\
        Group Normalization Numbers & 8 \\
        Activation (ResNet, Time MLP) & SiLU, GELU \\
        Downsampling/Upsampling Steps & 2 \\
        \hline
    \end{tabular}
    \vspace{-4mm}
\end{table}

\subsection{System Configurations}
\paragraph{Dataset Generation}
We evaluate the performance of the proposed RC-Flow using the 3GPP clustered delay line (CDL) channel models, which represent realistic propagation environments for 5G/6G MIMO systems \cite{3gpp.38.901}. The dataset comprises four distinct channel profiles: CDL-A, CDL-B, CDL-C, and CDL-D. Specifically, CDL-D represents a line-of-sight (LOS) environment, whereas CDL-B and CDL-C represent non-line-of-sight (NLOS) scenarios. CDL-A contains a mixture of both LOS and NLOS components. For each channel model, we generate 10,000 independent channel realizations for training and 100 realizations for testing. Furthermore, a mixed dataset is constructed by aggregating the training samples from all four CDL profiles.

\paragraph{Network Architecture}
The CFM network employs a time-conditioned U-Net architecture, and the real and imaginary channel components are processed as a two-channel real-valued input. The backbone follows an encoder-decoder structure with multi-level ResNet blocks \cite{he2016deep}, whose specific configurations are detailed in Table \ref{tab:hyperparams}. Temporal information $t$ is integrated into each block via sinusoidal positional encoding and a scale-and-shift mechanism.

\paragraph{Training and Inference}
The model is optimized via the CFM objective using the Adam optimizer. To ensure robustness across diverse link conditions, we define the probability path between the clean ground-truth channel $\mathbf{H}_0$ and a noisy endpoint $\mathbf{H}_1 = \mathbf{H}_0 + \mathbf{N}$, where the variance of $\mathbf{N}$ is governed by training SNRs. These training SNRs are uniformly sampled from $\mathcal{S}_{\text{dB}} \in \{-10, -5, \dots, 30\}$ dB for each batch. This SNR-aware strategy allows the network to learn a universal vector field capable of rectifying trajectories from various noise levels back to the clean manifold. Additionally, the time step $t$ is sampled from a logit-normal distribution, defined as $t = \sigma(\xi)$ where $\xi \sim \mathcal{N}(0, 1)$. Furthermore, an exponential moving average (EMA) is maintained to stabilize the weights, with comprehensive configurations detailed in Table \ref{tab:hyperparams}.
\par
During inference, $N_2$ is adaptively scheduled based on the noise standard deviation to balance estimation accuracy and computational complexity, as follows:
\begin{equation}
    N_2 = N_{\min} + \Delta N, 
    \label{eq:Inner Loop Steps}
\end{equation}
\begin{equation}
    \Delta N = \left\lfloor (N_{\max}-N_{\min}) \cdot \left( \frac{\log_{10}(\sigma_{\max} / \sigma_\text{pilot})}{\log_{10}(\sigma_{\max} / \sigma_{\min})} \right)^2 \right\rfloor,
\end{equation}
where $N_{\max}$ and $N_{\min}$ represent the predefined upper and lower bounds of the inner loop iteration numbers, while $\sigma_{\max}$ and $\sigma_{\min}$ denote the noise standard deviation boundaries that delineate the dynamic range of the environment. 



\begin{figure*}[t!]\vspace{0mm}
	\begin{center}
		\centerline{\includegraphics[width=0.9\textwidth]{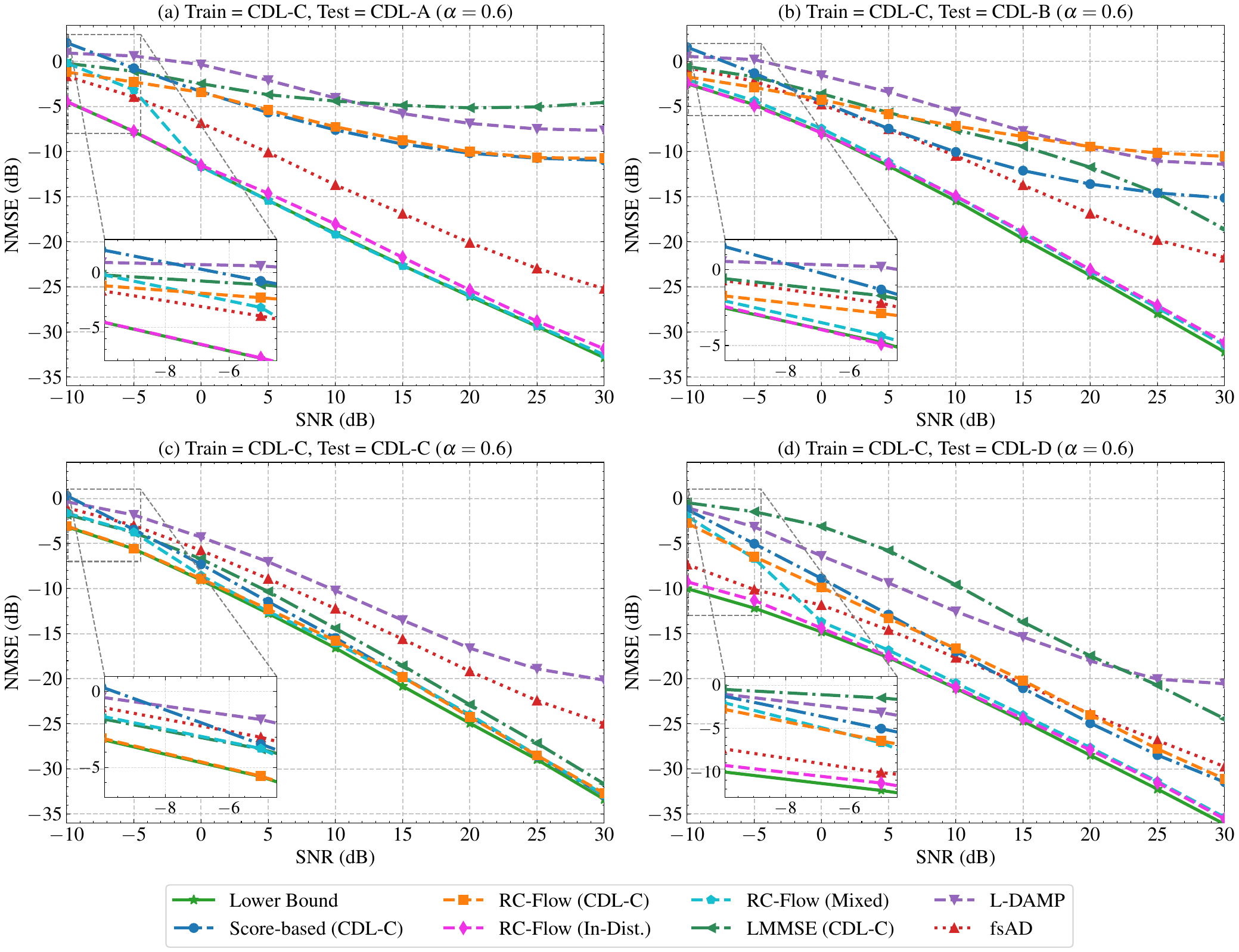}}  \vspace{-0mm}
		\captionsetup{font=footnotesize, name={Fig.}, labelsep=period} 
		\caption{\, Performance comparison of various schemes trained on CDL-C channels and evaluated across diverse channel conditions in a $16 \times 64$ mmWave MIMO system when $\alpha=0.6$.}
		\label{fig:cross scenarios} \vspace{-6mm}
	\end{center}
\end{figure*}

\paragraph{Hardware and Software Environment}
All experiments were implemented in Python 3.10 using the PyTorch 2.6.0 on the Ubuntu 22.04 system. The training and inference processes were accelerated using an NVIDIA A100 GPU. 

\subsection{Baselines}
To validate the effectiveness of the proposed method, we compare it against several baselines, ranging from classical estimation theory to recent deep learning-based approaches. The following baselines are evaluated in the simulations:

\begin{table}[t]
    \centering
    \caption{Computational Complexity, Latency, and Parameter Comparison}
    \label{tab:latency}
    \setlength{\tabcolsep}{4pt} 
    \renewcommand{\arraystretch}{1.1} 
    \begin{tabular}{l c c c}
        \toprule
        \textbf{Method} & \textbf{FLOPs (G)} & \textbf{Latency (ms)} & \textbf{Parameters (M)} \\
        \midrule
        LMMSE [10] & $5.5 \times 10^{-4}$ & $4.0 \times 10^{-1}$ & N/A \\
        fsAD [11] & $1.8 \times 10^{1}$ & $8.9 \times 10^{2}$ & N/A \\
        L-DAMP [14] & $1.5 \times 10^{0}$ & $3.5 \times 10^{-1}$ & $4.81 \times 10^{0}$ \\
        Score-based [34] & $1.1 \times 10^{4}$ & $5.0 \times 10^{2} $ & $5.89 \times 10^{0}$ \\
        \midrule
        \textbf{RC-Flow} & $4.3 \times 10^{1} $ & $ 2.28 \times 10^{0}$ & $3.88 \times 10^{0}$ \\
        \bottomrule
    \end{tabular}
    \begin{flushleft}
        \footnotesize{\textit{Note: FLOPs and latency are reported as the averaged time per sample, measured with a batch size of 100 on an NVIDIA A100 GPU.}}
    \end{flushleft}
    \vspace{-8mm}
\end{table}

\begin{itemize}
    \item \textbf{Linear Minimum Mean Square Error (LMMSE)} \cite{nayebi2017semi}: This estimator represents the optimal linear strategy designed to minimize the Bayesian MSE between the estimated and true channel matrices. By incorporating prior statistical knowledge, specifically the channel spatial correlation matrix, it effectively balances the reliability of the received observation against the prior distribution of the channel. The formulation leverages the second-order statistics of the channel, expressed as:
    \begin{equation}
    \mathbf{H}_{\text{LMMSE}} = \mathbf{Y} (\mathbf{P}^H \mathbf{R}_\mathbf{H} \mathbf{P} + \sigma^2_{\text{pilot}} \mathbf{I})^{-1} \mathbf{P}^H \mathbf{R}_\mathbf{H}
    \end{equation}
    where $\mathbf{R}_\mathbf{H}$ denotes the channel covariance matrix representing the spatial correlation.
    
    \item \textbf{Frequency-selective Atomic Decomposition (fsAD)} \cite{bhaskar2013atomic}: 
    Representing the class of CS methods, the fsAD algorithm exploits the inherent sparsity of millimeter-wave channels in the continuous angular domain. It formulates channel estimation as an atomic norm minimization problem. While effective for sparse channels, its performance may degrade in rich scattering environments where the channel exhibits low sparsity or when the number of propagation paths increases significantly, challenging the fundamental CS assumptions.

    \item \textbf{Learned Denoising-based Approximate Message Passing (L-DAMP)} \cite{metzler2017learned}: 
    A model-based deep learning approach that unfolds the Approximate Message Passing (AMP) algorithm into a neural network architecture. By replacing the traditional shrinkage function in AMP with a learned denoiser (typically a DnCNN \cite{zhang2017beyond}), L-DAMP effectively learns the channel prior from training data to reconstruct sparse signals from noisy measurements.

    \item \textbf{Score-based} \cite{arvinte2022mimo}: 
    A generative approach based on the noise conditional score network. This method learns the gradient of the log-prior distribution (the score function) of the channel. During inference, it employs annealed Langevin dynamics to iteratively denoise a random initialization, producing a single sample from the posterior distribution $p(\mathbf{h}|\mathbf{y})$. This represents the performance of posterior sampling.

    \item \textbf{Lower Bound} \cite{arvinte2022mimo}: 
    An enhanced inference strategy utilizing the same pre-trained Score-based model as above \cite{arvinte2022mimo}. 
    To establish a performance lower bound, this strategy approximates the MMSE estimator using the same pre-trained Score-based model. Instead of relying on a single posterior sample, it computes the posterior mean $\mathbb{E}[\mathbf{h}|\mathbf{y}]$ by averaging 50 independent samples generated from parallel Langevin dynamics chains, thereby minimizing stochastic variance and providing the theoretically optimal error floor for the score-based approach.
\end{itemize}

\subsection{Complexity Analysis}
Table \ref{tab:latency} details the complexity and efficiency of the algorithms. While LMMSE and L-DAMP achieve sub-millisecond speeds, their NMSE performance is insufficient. Conversely, Score-based and fsAD methods exhibit prohibitive latencies ($10^2$--$10^3$ ms) and high computational loads ($10^4$ G FLOPs), exceeding typical channel coherence times. In contrast, RC-Flow ($N_1=1, N_2=50$) strikes a superior balance, delivering near-optimal NMSE at the $10^0$ ms magnitude. To demonstrate versatility, RC-Flow was tested across hardware tiers: it achieves $\sim$3 ms on professional GPUs (RTX 6000 Ada/3090) and remains manageable at $\sim$50 ms on modern CPUs (EPYC 9654/Ryzen 9700X). These results confirm that RC-Flow is hardware-versatile and suitable for real-time deployment.

\subsection{CSI Estimation Evaluation}
Fig. \ref{fig:cross scenarios} illustrates the robustness of the proposed RC-Flow in CDL-\{A,B,C,D\} environments, employing a $16 \times 64$ MIMO configuration with $\alpha = 0.6$. For comparison, a Mixed RC-Flow model is trained on a composite dataset of CDL-\{A,B,C,D\} comprising 40,000 channel realizations. Furthermore, the performance of in-distribution models trained specifically on each corresponding CDL model is included to characterize the estimation accuracy within matched environments. As for performance measurement, we use normalized mean squared error (NMSE), defined as:
\begin{equation}
    \text{NMSE}= 10\log_{10}\frac{{\|\Hmat_{\text{est}} - \Hmat\|}^2_F}{{\|\Hmat\|}^2_F}
\end{equation}
\par
In in-distribution scenarios, the proposed RC-Flow outperforms the score-based baseline and closely approaches the lower bound. While RC-Flow exhibits nearly identical precision to the score-based method in high-SNR regimes, it achieves a substantial performance gain of approximately 2 dB in the low-SNR region ($-10$ dB to $-5$ dB). This superior noise resilience underscores the potential of RC-Flow as a high-fidelity estimation framework, particularly for power-limited environments. Although the LMMSE estimator is competitive across all SNR regimes, it consistently lags behind the proposed RC-Flow by a margin exceeding 2 dB.
\par
\begin{figure}[t]\vspace{-2mm}
	\begin{center}
		\centerline{\includegraphics[width=0.45\textwidth]{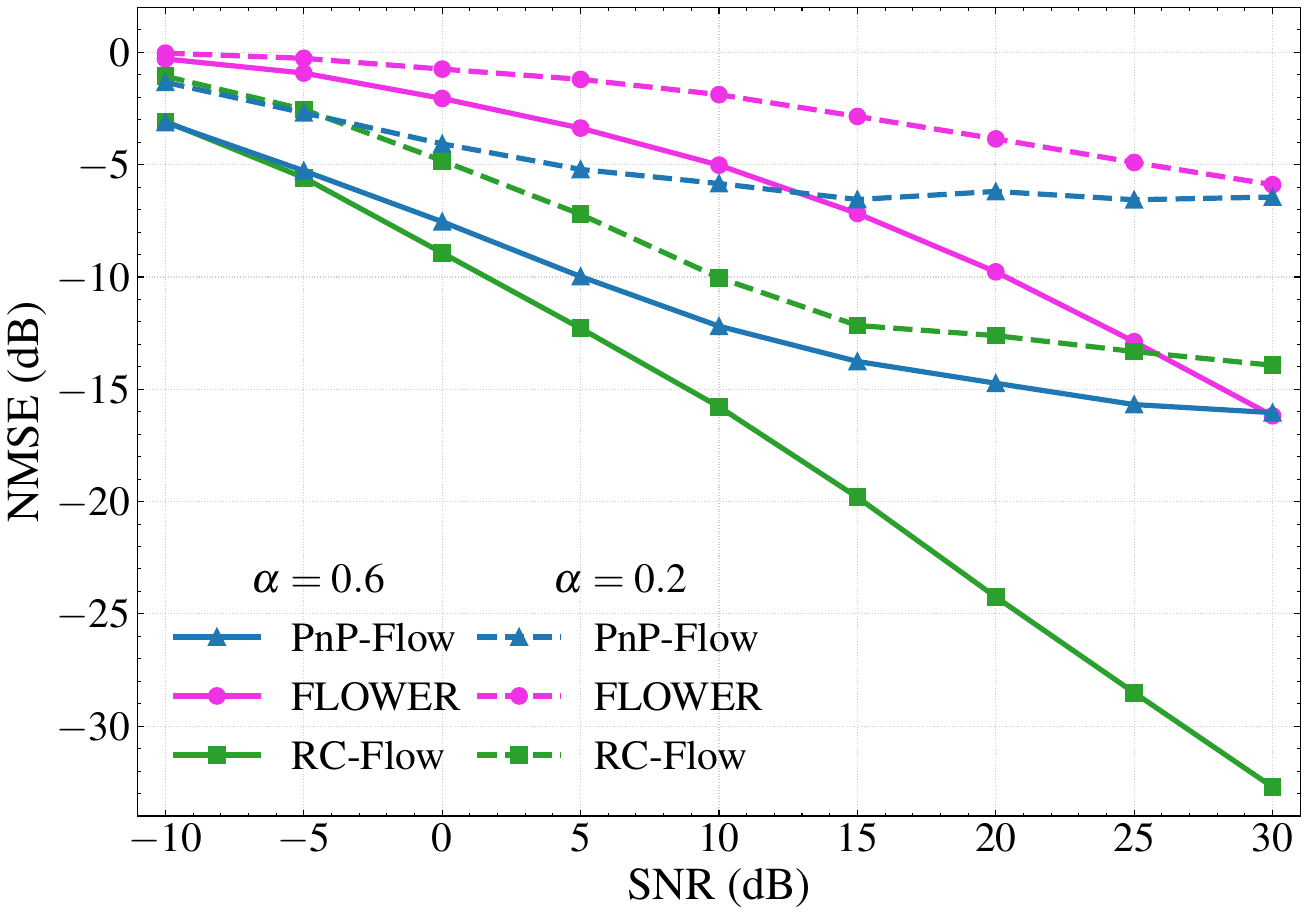}}  \vspace{-0mm}
		\captionsetup{font=footnotesize, name={Fig.}, labelsep=period} 
		\caption{\, Performance comparison of flow-matching based solvers on CDL-C channels, $N_t=64$, and $N_r=16$.}
		\label{fig:flow solver} \vspace{-8mm}
	\end{center}
\end{figure}

\begin{figure}[t]\vspace{-2mm}
	\begin{center}
		\centerline{\includegraphics[width=0.45\textwidth]{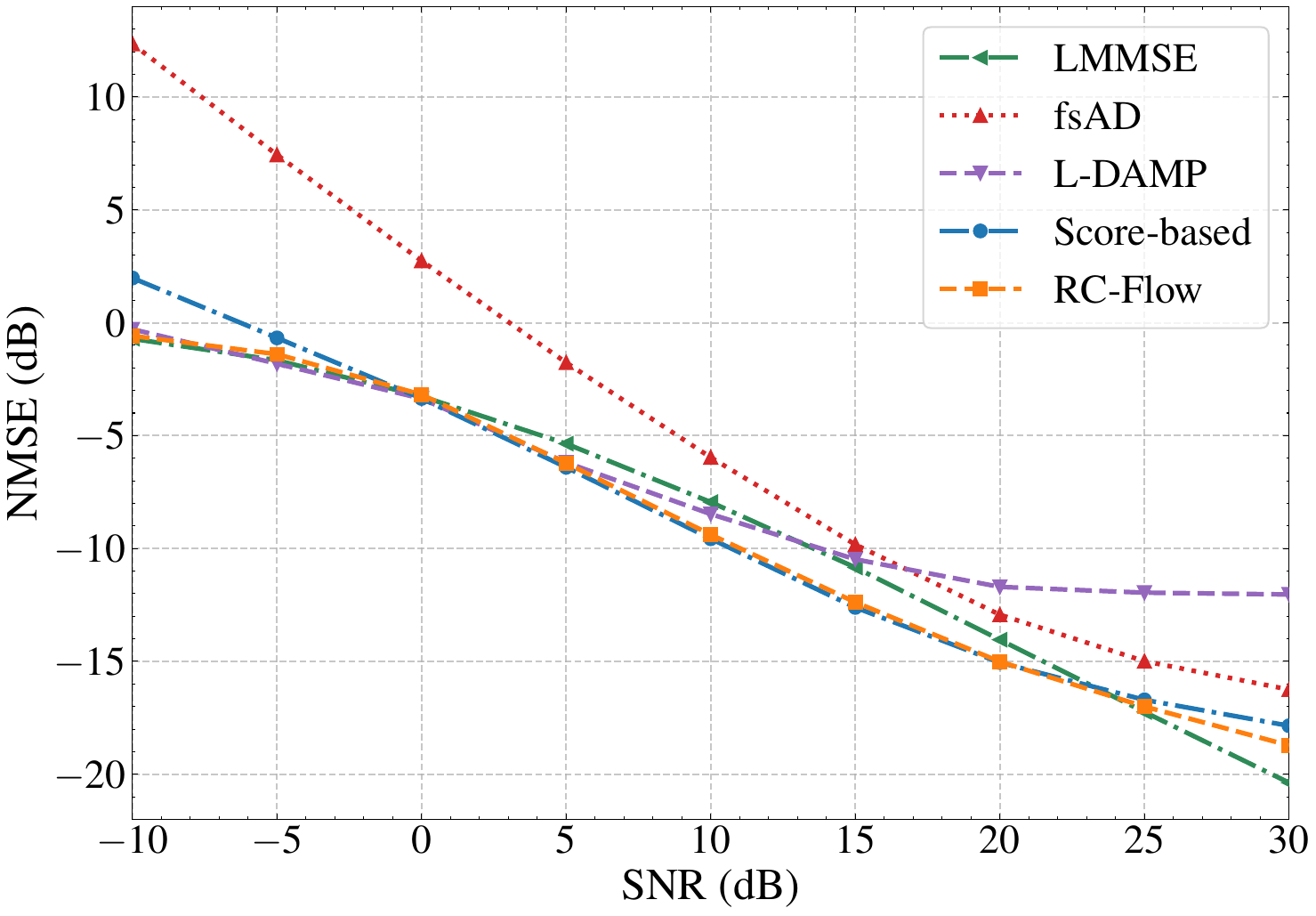}}  \vspace{-0mm}
		\captionsetup{font=footnotesize, name={Fig.}, labelsep=period} 
		\caption{\, Performance on the Sionna dataset at 40 GHz. The evaluation is conducted with pilot ratio $\alpha=0.6$, $N_t=64$, and $N_r=16$.}
		\label{fig:sionna_results} \vspace{-8mm}
	\end{center}
\end{figure}

\begin{figure}[t]\vspace{-2mm}
	\begin{center}
		\centerline{\includegraphics[width=0.45\textwidth]{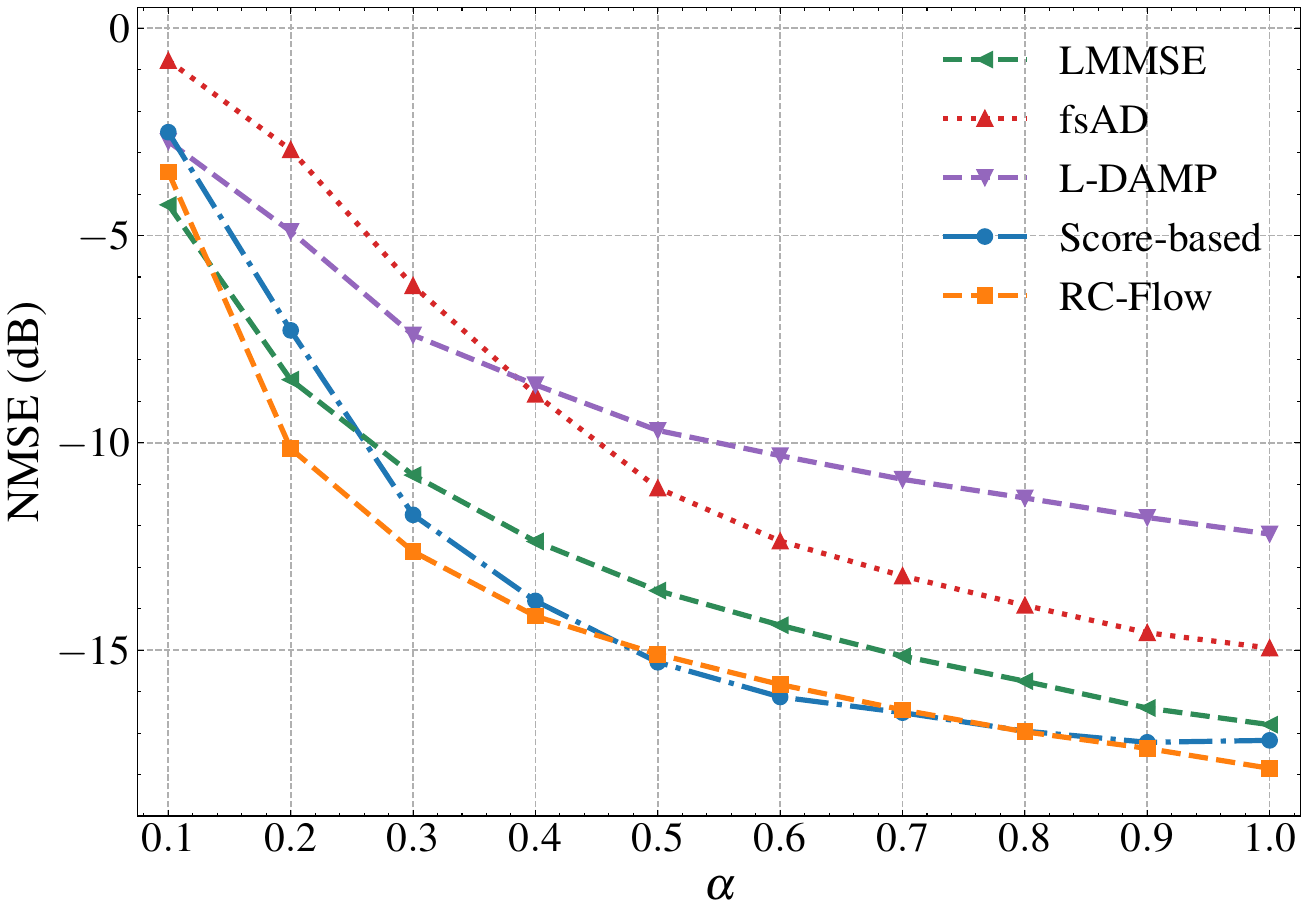}}  \vspace{-0mm}
		\captionsetup{font=footnotesize, name={Fig.}, labelsep=period} 
		\caption{\, Performance of different schemes trained on CDL-C channels and tested on CDL-D channels at an SNR of 10 dB with different pilot density.}
		\label{fig:different_pilot_density} \vspace{-8mm}
	\end{center}
\end{figure}

\begin{figure}[t]\vspace{-2mm}
	\begin{center}
		\centerline{\includegraphics[width=0.45\textwidth]{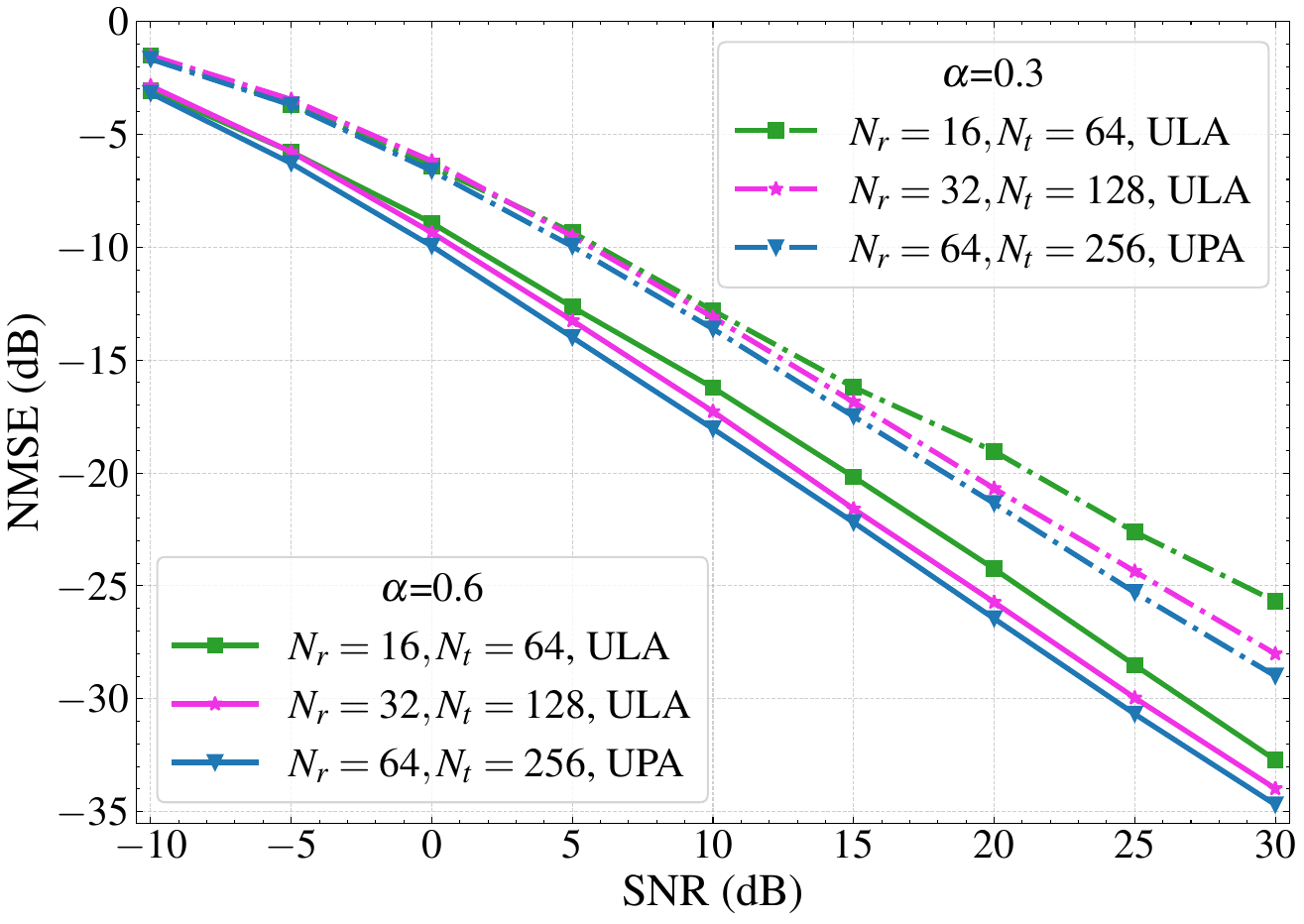}}  \vspace{-0mm}
		\captionsetup{font=footnotesize, name={Fig.}, labelsep=period} 
		\caption{\,  Performance of the proposed schemes trained and tested on CDL-C channels with different MIMO size.}
		\label{fig:MIMO_Antenna_Number} \vspace{-8mm}
	\end{center}
\end{figure}

\par
The generalization capabilities of the generative priors are further analyzed through the Mixed model and out-of-distribution (OOD) evaluations. The Mixed model maintains robustness at high SNR but suffers pronounced NMSE degradation at low SNR (excluding CDL-B) relative to its in-distribution counterparts. This stems from the structural ambiguity of a generalized prior, which lacks the specialization necessary to resolve profile-specific features when measurement guidance is weak. In OOD scenarios, the RC-Flow trained on CDL-C exhibits an NMSE floor of approximately $-11$ dB when evaluated on CDL-A and CDL-B. In contrast, negligible saturation is observed on CDL-D channels, where the NMSE decreases consistently with SNR, maintaining a 5 dB gap relative to the in-distribution benchmark. This discrepancy is attributed to channel complexity: the rich scattering in CDL-A and B increases estimation difficulty, whereas the dominant LOS components in CDL-D facilitate simpler recovery. 
\par
Finally, conventional baselines such as L-DAMP and fsAD provide relatively reliable and consistent performance across all scenarios, particularly for CDL-C and D. However, a critical limitation of fsAD is its strict dependence on precise antenna array geometry, which is often unavailable in practical deployments. Similarly, while Compressed Sensing (CS) based approaches align with the structural sparsity of CDL models, their reliance on idealistic sparsity assumptions often proves infeasible in complex real-world environments.

\begin{figure}[t]\vspace{-2mm}
	\begin{center}
		\centerline{\includegraphics[width=0.45\textwidth]{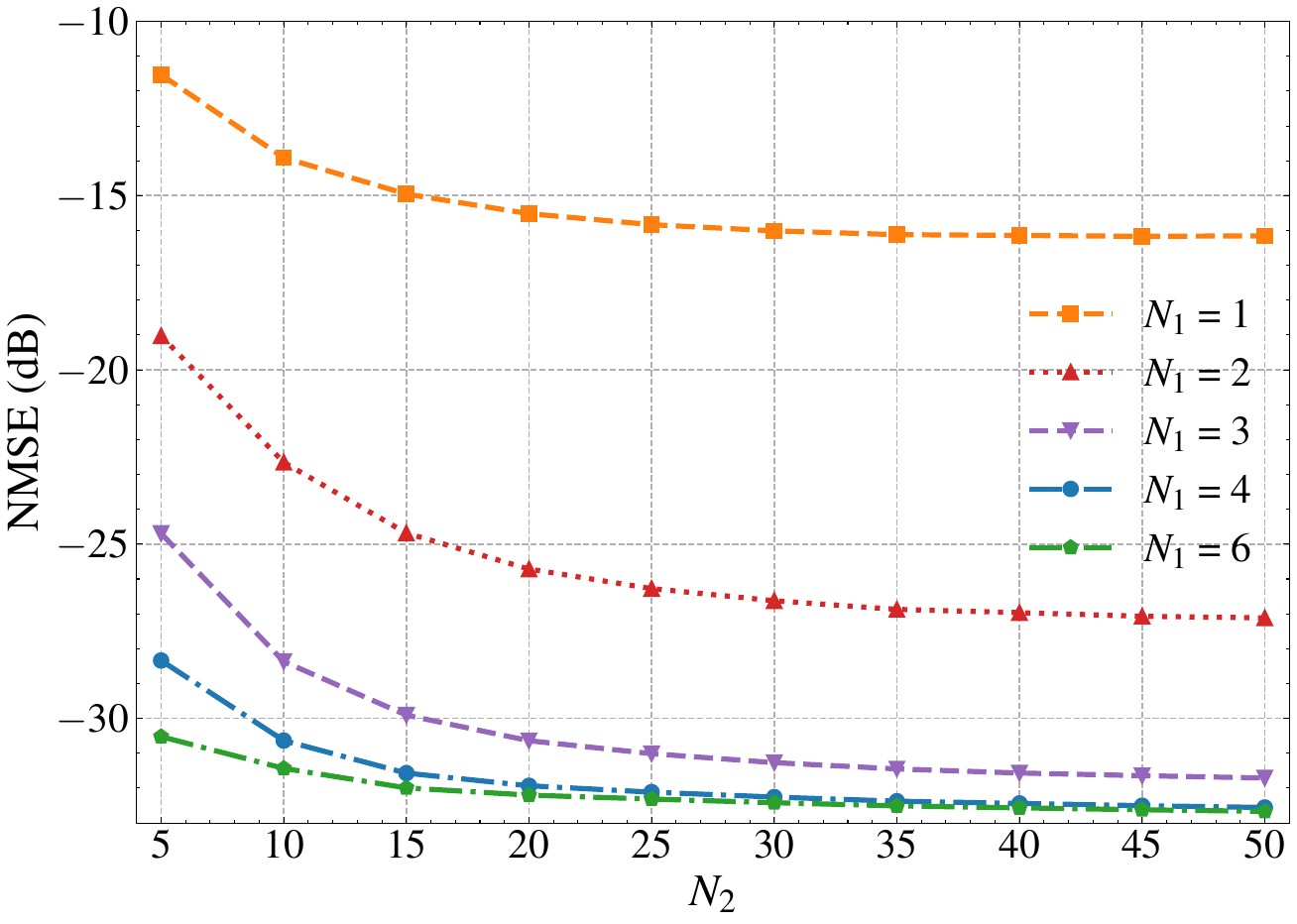}}  \vspace{-0mm}
		\captionsetup{font=footnotesize, name={Fig.}, labelsep=period} 
		\caption{\,  Performance of the proposed scheme trained and tested on CDL-C channels at an SNR of 30 dB with different number of inner ($N_2$) and outer loop iterations ($N_1$), when $\alpha=0.6$.}
		\label{fig:different_nfe1_nfe2} \vspace{-8mm}
	\end{center}
\end{figure}

\subsection{Performance Benchmarking and Generalization Analysis}
\paragraph{Comparison to other flow solvers}
Fig. \ref{fig:flow solver} shows the performance comparison of flow-matching based solvers on CDL-C channels under $\alpha \in \{0.2, 0.6\}$. It is evident that the proposed RC-Flow consistently outperforms FLOWER and PnP-Flow, across the entire SNR range. In the low-SNR regime, the performance gap among the three methods is relatively modest, as the heavy measurement noise limits the refinement capability of generative priors. However, as the SNR increases, the advantages of RC-Flow become remarkably pronounced. For $\alpha = 0.6$, while both PnP-Flow and FLOWER exhibit an NMSE floor around $-16$ dB at an SNR of 30 dB, RC-Flow continues to descend sharply, achieving a remarkable NMSE of $-33$ dB. A similar trend is observed for the sparse pilot case ($\alpha = 0.2$), where RC-Flow maintains a robust lead of approximately $8$ dB over its counterparts at high SNR. This significant margin underscores the efficacy of our recursive anchor mechanism and anchored trajectory rectification, which effectively prevent the trajectory drift and error accumulation that plague conventional open-loop flow solvers in high-fidelity reconstruction tasks.

\paragraph{Performance on Sionna Dataset}
Fig. \ref{fig:sionna_results} evaluates the NMSE performance using ray-traced channels generated by the Sionna \cite{hoydis2022sionna,hoydis2023sionna} simulator at 40 GHz. While all other system parameters and default configurations remain unchanged with a pilot density of $\alpha = 0.6$. The results demonstrate the superior robustness of the proposed RC-Flow. Specifically, RC-Flow achieves excellent performance across the entire SNR spectrum, consistently outperforming the score-based baseline in both low- and high-SNR regimes. In contrast, the fsAD method exhibits mediocre performance across all SNR levels. While L-DAMP shows competitive results at low SNR, it suffers from a distinct performance saturation as the SNR increases. Furthermore, although the LMMSE estimator is competitive at extreme SNR points, it exhibits a noticeable performance dip in the mid-SNR range. These findings validate that the proposed framework delivers robust and generalized results across different channel environments.

\begin{figure}[t]\vspace{-1mm}
    \centering
    \captionsetup{font=footnotesize, name={Fig.}, labelsep=period}
    
    \begin{subfigure}[b]{0.45\textwidth}
        \centering
        \includegraphics[width=\textwidth]{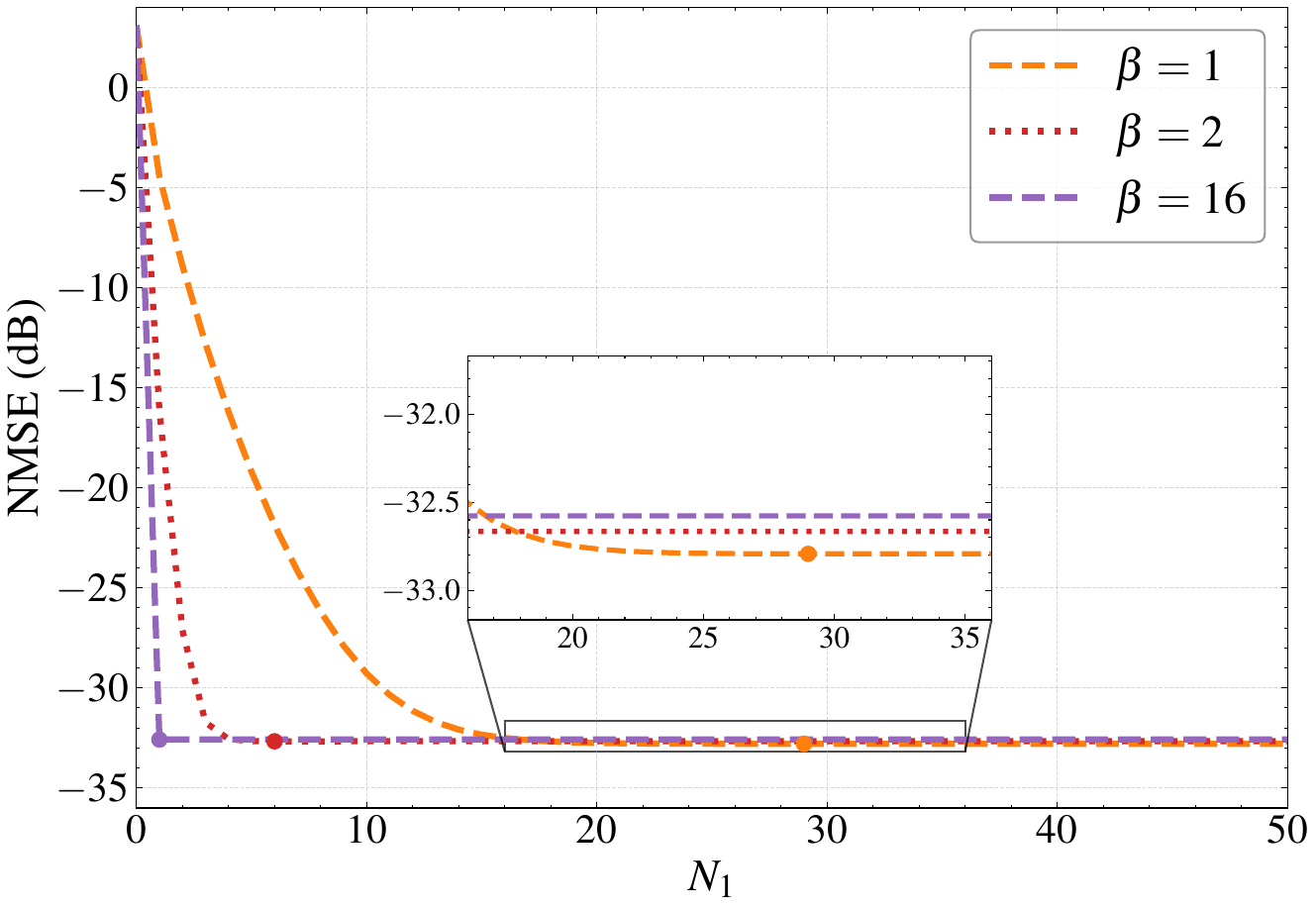}
        \caption{SNR = 30 dB} 
        \label{fig:Beta_N1_30dB}
    \end{subfigure}
    \hspace{5mm}
    \begin{subfigure}[b]{0.45\textwidth}
        \centering
        \includegraphics[width=\textwidth]{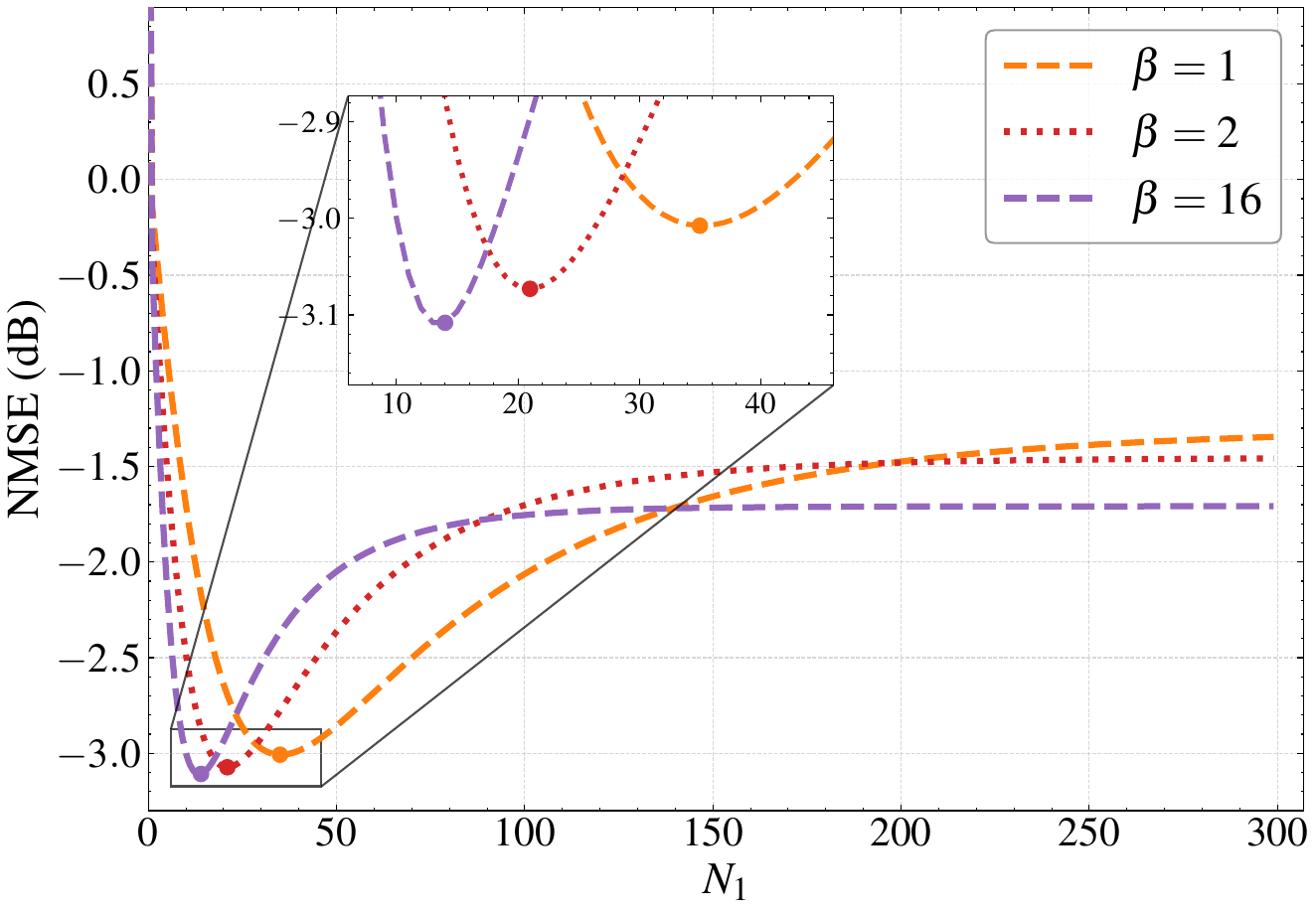}
        \caption{SNR = -10 dB} 
        \label{fig:Beta_N1_-10dB}
    \end{subfigure}
    
    \vspace{2mm} 
    \caption{The iterative process of the algorithm's dynamics with different $\beta$ and $\alpha=0.6$.}
    \label{fig:Beta_combined}
    \vspace{-8mm} 
\end{figure}

\begin{figure*}[t]\vspace{0mm}
	\begin{center}
		\centerline{\includegraphics[width=0.95\textwidth]{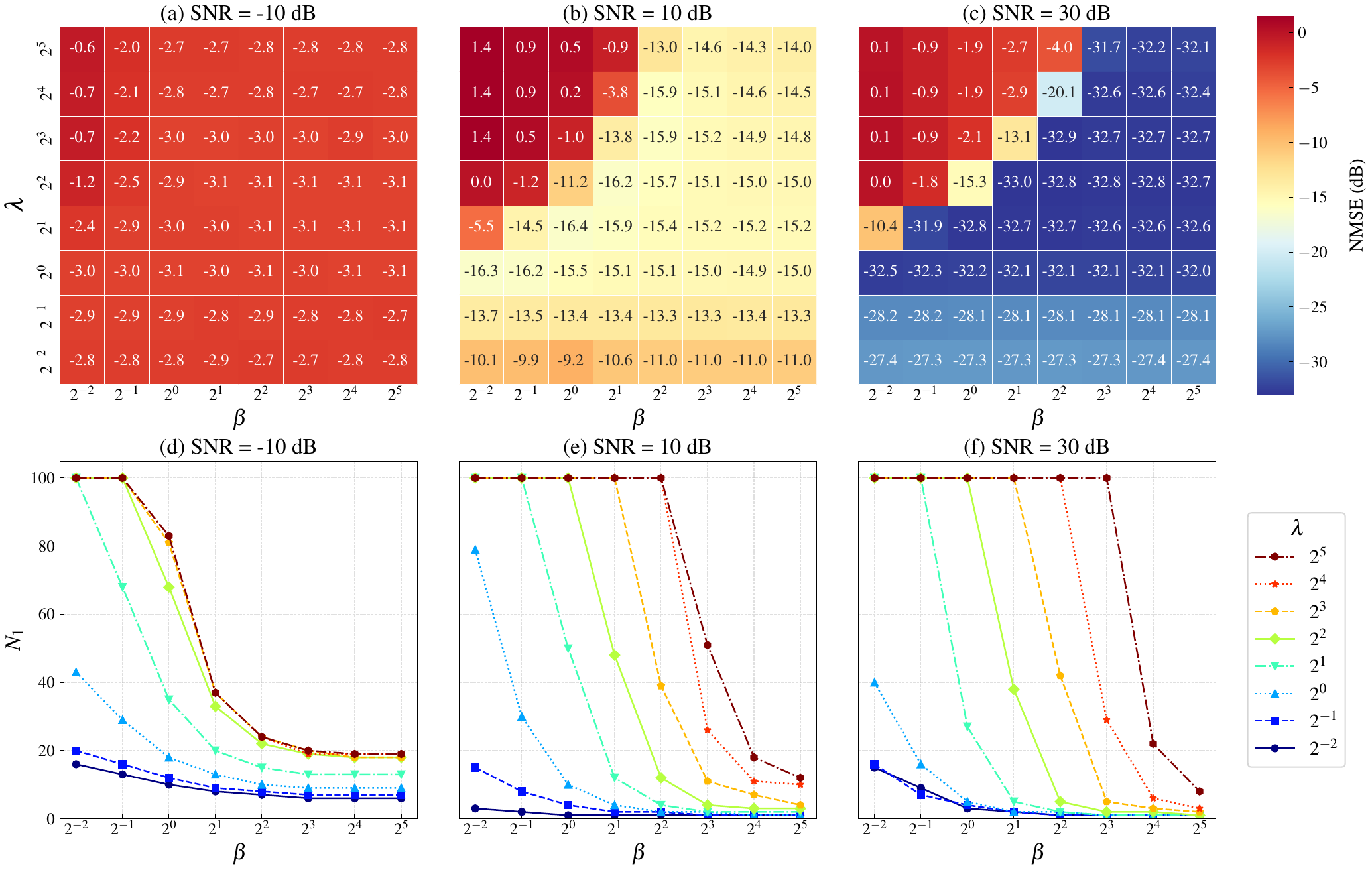}}  \vspace{-0mm}
		\captionsetup{font=footnotesize, name={Fig.}, labelsep=period} 
		\caption{\, Sensitivity analysis of hyperparameters $\lambda$ (Flow Schedule) and $\beta$ (Anchor Schedule) across different SNR regimes. The heatmap displays the NMSE (dB), and the line graph displays the least number of outer iterations required to reach the sweet spot as shown in Fig. \ref{fig:Beta_combined}.}
		\label{fig:hyper schedule} \vspace{-8mm}
	\end{center}
\end{figure*}

\subsection{System Scaling Analysis}
\paragraph{Performance Scaling with Pilot Overhead}
Fig. \ref{fig:different_pilot_density} illustrates the NMSE performance as a function of the pilot density $\alpha \in \{0.1, 0.2, \dots, 1.0\}$. All evaluated methods demonstrate consistent NMSE reduction as $\alpha$ increases, with the proposed RC-Flow achieving a significant performance advantage. Specifically, the LMMSE estimator demonstrates robust performance across all pilot densities, showing particular resilience in the extremely low pilot density regime ($\alpha=0.1$). However, it still lags significantly behind both the RC-Flow and Score-based methods in the medium-to-high regimes. In terms of L-DAMP, while it provides competitive results at very low pilot densities, its performance scales poorly with $\alpha$, resulting in marginal gains in the medium-to-high density regimes. Conversely, the fsAD method is more sensitive to pilot density, showing poor performance at low $\alpha$ but improving faster as $\alpha$ increases, though a substantial gap persists relative to the generative solvers. Notably, while RC-Flow and the Score-based method achieve similar precision for $\alpha > 0.4$, RC-Flow demonstrates a noticeable advantage in the low-pilot regime, outperforming the Score-based baseline by $3$ dB at $\alpha=0.2$. These findings underscore the remarkable resilience of RC-Flow to limited pilot overhead.

\paragraph{Scalability Across Massive MIMO Dimensions}
Fig. \ref{fig:MIMO_Antenna_Number} illustrates the NMSE performance of RC-Flow across three distinct MIMO configurations with $\alpha \in \{0.3, 0.6\}$, while maintaining antenna spacing at half-wavelength. For each antenna dimension, we train a separate model individually. As observed, the estimation accuracy improves consistently as the MIMO dimensions scale up. This performance gain is primarily attributed to the superior angular resolution and the more pronounced structural sparsity inherent in larger antenna arrays. These characteristics allow the generative prior to provide more precise guidance within high-dimensional channel manifolds.

\subsection{Impact of Iteration Numbers}\label{sec: Hyperparameter Analysis}
Fig. \ref{fig:different_nfe1_nfe2} illustrates the NMSE performance across varying outer and inner iteration numbers ($N_1$ and $N_2$) at an SNR of 30 dB with $\alpha=0.6$. For a fixed $N_1$, the NMSE decreases significantly as $N_2$ increases, typically plateauing beyond $N_2 = 25$. Notably, the marginal gains from increasing $N_2$ are more pronounced at lower $N_1$ values, whereas this influence diminishes as $N_1$ becomes larger. For a fixed $N_2$, increasing $N_1$ leads to significant improvements when $N_1$ is small, with gains becoming marginal once $N_1$ exceeds approximately 4 steps. These observations validate the efficacy of the proposed recursive anchor refinement mechanism. However, it is evident that both the outer processes and the inner fine-grained steps exhibit marginal effect, as the performance improvements gradually plateauing.

\subsection{Analysis of Algorithm Dynamics}
\paragraph{Convergence Dynamics and $\beta$ Tuning}

Figs. \ref{fig:Beta_N1_30dB} and \ref{fig:Beta_N1_-10dB} illustrate the convergence dynamics of the proposed algorithm with $\alpha=0.6$ and $\lambda=2$ under different SNR regimes.
\par
In the high SNR case, the NMSE exhibits a monotonic decreasing trend with respect to $N_1$, eventually stabilizing at a convergence floor. Increasing $\beta$ significantly accelerates convergence. For instance, at $\beta=16$, the algorithm achieves near-optimal performance within a single iteration. However, this acceleration incurs a marginal performance loss of approximately 0.1 dB compared to lower $\beta$ settings. Conversely, at low SNR case, the dynamics follow a non-monotonic pattern: the NMSE initially descends to a minimum `sweet spot', followed by a gradual ascent toward an asymptotic plateau. Notably, as $\beta$ increases, the minimum point shifts toward the bottom-left, indicating simultaneous improvements in inference accuracy and speed. Furthermore, larger $\beta$ values accelerate the transition to the plateau phase while achieving a lower converged NMSE.
\par

A comparative analysis reveals that while the algorithm shares a common initialization dynamics, its asymptotic behavior is strictly governed by the signal quality. In the initial phase, both regimes exhibit a rapid descent in NMSE dominated by noise suppression, where the model effectively steers the estimate toward the inherently smooth posterior mean. However, subsequent behaviors diverge due to the perception-distortion trade-off \cite{blau2018perception}. At $-10$ dB, the NMSE rebounds after reaching a minimum `sweet spot' as the model proceeds to synthesize high-frequency structural details to satisfy manifold constraints. However, due to the ill-posed nature of the problem, these generated stochastic features suffer from spatial misalignment, effectively prioritizing perceptual consistency at the expense of pixel-wise fidelity. Conversely, as the SNR increases to 30 dB, reliable measurements ensure that the synthesized fine-grained textures remain faithful to the true channel state. Consequently, the convergence region aligns with the sweet spot, avoiding the distortion penalty associated with hallucinatory features and resulting in the strictly monotonic trajectory where no early stopping is required.

\paragraph{Sensitivity analysis of $\lambda$-$\beta$ configurations}
\label{para:beta_and_lambda}
Fig. \ref{fig:hyper schedule} illustrates a comprehensive sensitivity analysis of RC-Flow regarding hyperparameters $\lambda$ and $\beta$ at $\alpha = 0.6$. The heatmaps (a) -- (c) represent the NMSE at the sweet spot, while the line graphs (d) -- (f) depict the outer iteration numbers required to reach this point within the maximum limit of $N_1 = 100$.
\par
Analysis of the heatmaps and line graphs reveal a distinct correspondence between reconstruction fidelity and convergence efficiency. Regarding NMSE performance, optimal results are consistently achieved in the moderate $\lambda$ regime, while larger $\lambda$ values significantly decelerate the convergence process. This occurs because an excessively high $\lambda$ causes the time variable $t$ to decrease too rapidly, diminishing the flow-predicted velocity and resulting in smaller per-step updates. Furthermore, while increasing $\beta$ generally accelerates convergence, enabling RC-Flow to reach the sweet spot in a single $N_1$ step under medium-to-high SNR scenarios, it induces a slight performance degradation of approximately 1 dB in medium SNR regimes. 
In the low-SNR regime, at least $N_1 = 6$ iterations are required under an optimal $\lambda$-$\beta$ configurations. This suggests that higher $\beta$ values prioritize convergence speed, potentially at the expense of absolute precision in intermediate noise conditions.
\par
Finally, we highlight that certain $\lambda$-$\beta$ configurations, particularly in the upper-left regions of the heatmaps, exhibit a marked decrease in NMSE performance. This sharp degradation does not necessarily imply that such $\lambda$-$\beta$ combinations result in poor potential performance, but rather that the trajectory has not yet reached the sweet spot within the maximum $N_1$ limit.  This underscores the significance of balanced scheduling to optimize both reconstruction fidelity and computational efficiency.

\section{Conclusion}\label{sec:conclusion}
In this paper, we presented RC-Flow, a unified generative framework for high-dimensional MIMO channel estimation. By reformulating the inference process as a closed-loop fixed-point iteration, our method addresses the fundamental limitations of open-loop generative samplers in ill-posed inverse problems. The core novelty lies in the serial restart mechanism combined with anchored trajectory rectification, which effectively suppresses error propagation, even under extreme noise conditions. This design allows RC-Flow to achieve SOTA accuracy across a wide range of SNRs and pilot density, reducing inference latency by orders of magnitude compared to score-based approach. We believe this framework establishes a new paradigm for solving constrained inverse problems, paving the way for low-latency, high-fidelity generative inference in 6G communications.

\begin{appendices}

\renewcommand{\thefigure}{\thesection\arabic{figure}}

\vspace{2mm}

\section{Proof of Theorem 1}
\setcounter{figure}{0}
\label{app:proof_existence}

To invoke Brouwer's Fixed-Point Theorem, we must identify a compact convex ball $\mathcal{K}_{R^\star}$ such that $\mathcal{T}(\mathcal{K}_{R^\star}) \subseteq \mathcal{K}_{R^\star}$.
\begin{figure*}[t]\vspace{-2mm}
	\begin{center}
		\centerline{\includegraphics[width=0.95\textwidth]{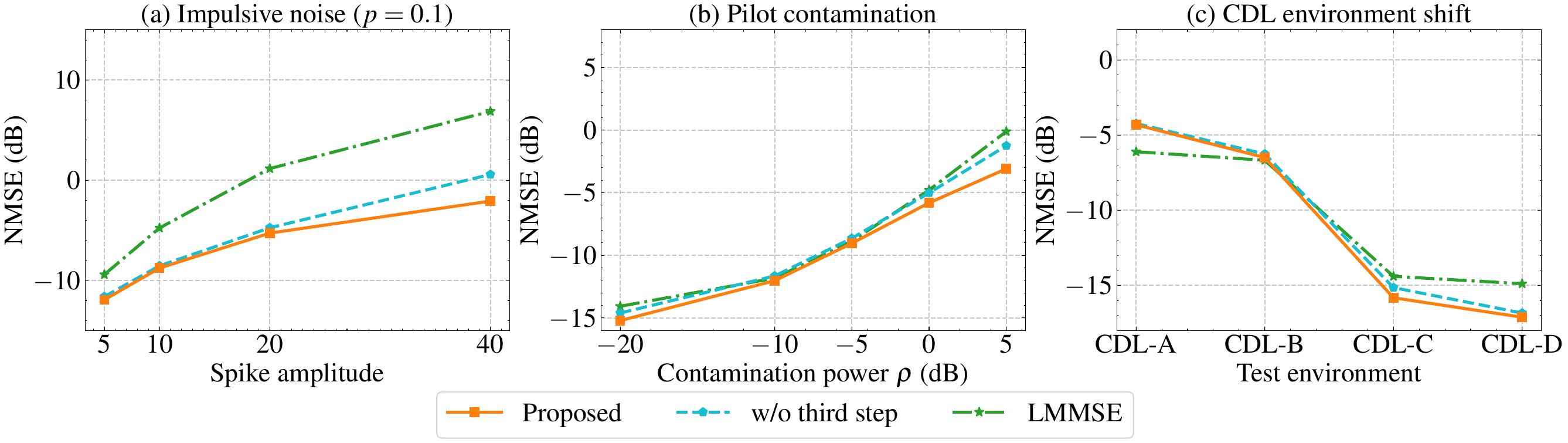}}  \vspace{-0mm}
		\captionsetup{font=footnotesize, name={Fig.}, labelsep=period} 
		\caption{\, Empirical stability under challenging edge cases. 
(a) Robustness to impulsive noise spikes with increasing anomaly intensity. 
(b) Stability under pilot contamination with increasing interference strength. 
(c) Generalization under CDL environment shifts across CDL-A/B/C/D.}
		\label{fig:Bounded_Denoiser_Stress} \vspace{-8mm}
	\end{center}
\end{figure*}
\subsection{Continuity}
The neural network $\mathcal{D}$, the linear matrix operations in the projection $\mathcal{P}$, and the convex combination with the anchor are all continuous mappings. Since the matrix inverse $(\mathbf{M} + w^{-1}\mathbf{I})^{-1}$ exists and is continuous for $w^{-1} > 0$, the composite operator $\mathcal{T}$ is continuous on $\mathbb{C}^{N_r \times N_t}$.

\subsection{Self-Mapping on a Compact Set}
We establish a global upper bound on the projected states, independent of the iteration index. Based on Assumption \ref{assump:bounded}, the output is bounded by $\|\mathcal{D}(\mathbf{H})\|_F \le B_{\text{flow}}$. Consider the projection step at any inner index $i$:
\begin{equation}
    \mathbf{H}_{\text{proj}}^{(i)} = (\mathbf{R} + w^{-1} \mathcal{D}(\mathbf{H}^{(k, i)})) (\mathbf{M} + w^{-1} \mathbf{I})^{-1}.
\end{equation}
Applying the norm inequality and the spectral bound $\| w^{-1} (\mathbf{M} + w^{-1} \mathbf{I})^{-1} \|_2 \le 1$:
\begin{align}
    \|\mathbf{H}_{\text{proj}}^{(i)}\|_F &\le \|\mathbf{R}(\mathbf{M} + w^{-1} \mathbf{I})^{-1}\|_F + \|\mathcal{D}(\cdot)\|_F \cdot 1 \nonumber \\
    &\le C_{\text{obs}} + B_{\text{flow}} \triangleq R_{\text{limit}}.
\end{align}
Crucially, $R_{\text{limit}}$ is a constant determined solely by the observations and the physical prior manifold. It does not depend on the norm of the intermediate state $\mathbf{H}^{(k, i)}$ or the step index $i$.
\par
Define the invariant ball $\mathcal{K}_{R_{\text{limit}}} \triangleq \{ \mathbf{H} \mid \|\mathbf{H}\|_F \le R_{\text{limit}} \}$, which is inherently a convex set.
Assume the initialization satisfies $\mathbf{H}^{(k, 0)} \in \mathcal{K}_{R_{\text{limit}}}$, implying the anchor $\boldsymbol{\epsilon} \in \mathcal{K}_{R_{\text{limit}}}$.
Since we have established that any projected state is bounded by $R_{\text{limit}}$, we also have $\mathbf{H}_{\text{proj}}^{(i)} \in \mathcal{K}_{R_{\text{limit}}}$.
The update rule in \eqref{eq:rectification} constitutes a convex combination. By the definition of convexity, the resulting state $\mathbf{H}^{(k, i+1)}$ must remain within $\mathcal{K}_{R_{\text{limit}}}$ for any $t'_i \in [0, 1)$.
By induction, the final output satisfies $\mathbf{H}^{(k, N_2)} \in \mathcal{K}_{R_{\text{limit}}}$.
\par
Therefore, we have proved that if the input $\mathbf{H}^{(k, 0)} \in \mathcal{K}_{R^\star}$, then the output of the composite inner loop $\mathcal{T}(\mathbf{H}^{(k, 0)}) \in \mathcal{K}_{R^\star}$. Thus, $\mathcal{T}$ is a self-mapping on the compact convex set $\mathcal{K}_{R^\star}$. By Brouwer's Fixed-Point Theorem, there exists at least one fixed point $\mathbf{H}^\star \in \mathcal{K}_{R^\star}$. \hfill

\subsection{Empirical Stability and Edge-Case Analysis}
To evaluate the practical stability of the solver beyond the theoretical Assumption~\ref{assump:bounded}, we stress-test the CDL-C trained model against three edge cases: impulsive noise ($p=0.1$, spikes up to $40\sigma$), pilot contamination ($\rho \in [-20, 5]$ dB), and environment shifts (OOD testing on CDL-A/B/C/D). As shown in Fig.~\ref{fig:Bounded_Denoiser_Stress}, the proposed solver exhibits graceful degradation without numerical divergence, even when the bounded-denoiser hypothesis is challenged.
\par
Specifically, under extreme impulsive noise (spike amplitude 40), RC-Flow maintains an NMSE of $-2.09$ dB, significantly outperforming the variant without the third step ($0.57$ dB) and LMMSE ($6.87$ dB). In severe pilot contamination ($\rho=5$ dB), the proposed method achieves $-3.09$ dB NMSE, compared to $-1.25$ dB and $-0.11$ dB for the other two schemes. Furthermore, across diverse CDL environment shifts, the solver remains stable with the $95$th percentile output norm ratio remaining close to unity. These results confirm that the anchored iterative procedure provides empirical robustness against severe noise artifacts and scattering-environment mismatches, even without an explicit Lipschitz constraint on the U-Net.

\vspace{2mm}
\section{Proof of Theorem 2}
\setcounter{figure}{0}
\label{app:proof_stability}
In this appendix, we provide a rigorous derivation of the global asymptotic stability in a generalized metric space for the RC-Flow algorithm. The objective is to bound the spectral radius of the total Jacobian matrix $\mathbf{J}_{\mathcal{T}} \in \mathbb{C}^{N \times N}$ (where $N = N_t N_r$), which governs the sensitivity of the output of the $k$-th outer iteration, $\mathbf{H}^{(k+1, 0)}$, with respect to its input, $\mathbf{H}^{(k, 0)}$.

\subsection{Stability Criterion via Induced Metric Space}
We invoke Assumption \ref{assump:contraction}, which postulates that the operator $\mathbf{T}_i \triangleq \mathbf{J}_{\mathcal{P}, i}\mathbf{J}_{\mathcal{D}, i}$ satisfies the spectral contraction condition $\rho(\mathbf{T}_i) \le \gamma < 1$. According to the Householder theorem, for any matrix $\mathbf{A}$ with $\rho(\mathbf{A}) < 1$ and any given $\eta > 0$, there exists an induced vector norm $\|\cdot\|_*$ such that:
\begin{equation}
    \|\mathbf{A}\|_* \le \rho(\mathbf{A}) + \eta.
\end{equation}
However, since the operator $\mathbf{T}_i$ is time-varying, strictly establishing stability requires a common metric. Given the smooth evolution of the parameters, we analytically extend Assumption \ref{assump:contraction} to postulate the existence of a common induced norm $\| \cdot \|$ for the sequence $\{\mathbf{T}_i\}^{N_2-1}_{i=0}$, such that:
\begin{equation} \label{eq:induced_contraction}
    \|\mathbf{T}_i\|_* \le \gamma' < 1, \quad \forall i.
\end{equation}
Consequently, our proof strategy is to establish the contraction of the total Jacobian $\mathbf{J}_{\mathcal{T}}$ within this generalized metric space.

\subsection{Differential Dynamics of the Inner Loop}

Consider the $k$-th outer loop iteration. Let the input state be $\mathbf{H}^{(k, 0)}$. According to Algorithm 1, the anchor point is reset to the previous estimate, i.e., $\boldsymbol{\epsilon} = \mathbf{H}^{(k, 0)}$. The inner loop generates a sequence of states $\{\mathbf{H}^{(k, i)}\}_{i=0}^{N_2}$, where the update rule at the $i$-th inner step is given by:
\begin{equation}
    \mathbf{H}^{(k, i+1)} = t'_i \mathbf{H}^{(k, 0)} + (1-t'_i) \mathcal{P}(\mathcal{D}(\mathbf{H}^{(k, i)})),
\end{equation}
\par
We define the cumulative sensitivity matrix at inner step $i$ as $\mathbf{J}^{(i)} \triangleq \frac{\partial \rm{vec}(\mathbf{H}^{(k, i)})}{\partial \rm{vec}(\mathbf{H}^{(k, 0)})}$. By applying the chain rule to the recurrence relation and using the definition of the composite Jacobian $\mathbf{T}_i$, we obtain:
\begin{equation}
    \label{eq:recurrence_relation}
    \mathbf{J}^{(i+1)} = t'_i \mathbf{I} + (1-t'_i) \mathbf{T}_i \mathbf{J}^{(i)}.
\end{equation}
where $\mathbf{I}$ is the identity matrix resulting from differentiating the anchor term $t'\mathbf{H}^{(k, 0)}$.

\subsection{Recursive Expansion and Closed-Form Solution}

The base case for the recurrence is $\mathbf{J}^{(0)} = \mathbf{I}$. To reveal the structure of the total Jacobian, let us define the path transition operator $\mathbf{\Phi}_{m, j}$, which represents the accumulated decay and rotation from step $j$ to $m$:
\begin{equation}
    \mathbf{\Phi}_{m, j} \triangleq \begin{cases}
        \prod_{n=j}^{m-1} (1-t'_n) \mathbf{T}_n, & \text{if } m > j, \\
        \mathbf{I}, & \text{if } m = j.
    \end{cases}
\end{equation}
By induction, the total Jacobian $\mathbf{J}_{\mathcal{T}} \triangleq \mathbf{J}^{(N_2)}$ at the end of the inner loop can be expressed as a weighted sum of historical transition paths:
\begin{equation}
    \label{eq:closed_form_expansion}
    \mathbf{J}_{\mathcal{T}} = \mathbf{\Phi}_{N_2, 0} + \sum_{i=0}^{N_2-1} t'_i \mathbf{\Phi}_{N_2, i+1}.
\end{equation}

\subsection{Strict Contraction Analysis via Partition of Unity}

We now apply the induced norm $\|\cdot\|_*$ defined in \eqref{eq:induced_contraction} to the closed-form expansion \eqref{eq:closed_form_expansion}, we obtain:
\begin{equation} \label{norm_inequality}
    \|\mathbf{J}_{\mathcal{T}}\|_* \le \|\mathbf{\Phi}_{N_2, 0}\|_* + \sum_{i=0}^{N_2-1} t'_i \|\mathbf{\Phi}_{N_2, i+1}\|_*.
\end{equation}
The norm of the transition operator is bounded by the product of individual operator norms. Using $\|\mathbf{T}_n\|_* \le \gamma' < 1$, we have:
\begin{equation}
    \|\mathbf{\Phi}_{N_2, j}\|_* \le \prod_{n=j}^{N_2-1} (1-t'_n) \|\mathbf{T}_n\|_* \le (\gamma')^{N_2-j} \prod_{n=j}^{N_2-1} (1-t'_n).
\end{equation}
For brevity, we define the cumulative scalar decay factor $P_j$:
\begin{equation}
    P_j \triangleq \prod_{n=j}^{N_2-1} (1-t'_n), \quad \text{with } P_{N_2} \triangleq 1.
\end{equation}
Substituting this bound into the inequality \eqref{norm_inequality}, we arrive at:
\begin{equation}
    \label{eq:norm_bound_final}
    \|\mathbf{J}_{\mathcal{T}}\|_* \le P_0 (\gamma')^{N_2} + \sum_{i=0}^{N_2-1} t'_i P_{i+1} (\gamma')^{N_2-(i+1)}.
\end{equation}

To prove strict contraction, we utilize the algebraic property of the weighting coefficients, yielding:
\begin{equation}
    P_0 + \sum_{i=0}^{N_2-1} t'_i P_{i+1} = P_0 + \sum_{i=0}^{N_2-1} (P_{i+1} - P_i) = P_{N_2} = 1.
\end{equation}
Therefore, the upper bound in \eqref{eq:norm_bound_final} represents a convex combination of the terms $(\gamma')^k$.

Since $\gamma' < 1$ and $N_2 \ge 1$, the factor $(\gamma')^k$ is strictly less than $1$ for all $k \ge 1$. Specifically, the exponents in \eqref{eq:norm_bound_final} are $N_2, N_2-1, \dots, 0$. Unless $t'_{N_2-1}=1$, the sum contains terms scaled by $\gamma' < 1$. Thus, the strict inequality holds:
\begin{equation}
    \|\mathbf{J}_{\mathcal{T}}\|_* < P_0 \cdot 1 + \sum_{i=0}^{N_2-1} t'_i P_{i+1} \cdot 1 = 1.
\end{equation}
Since $\|\mathbf{J}_{\mathcal{T}}\|_* < 1$, the mapping is a strict contraction in the metric space $(\mathbb{C}^N, \|\cdot\|_*)$. Therefore, the proof is complete. \hfill
\begin{figure}[t]\vspace{-2mm}
	\begin{center}
		\centerline{\includegraphics[width=0.45\textwidth]{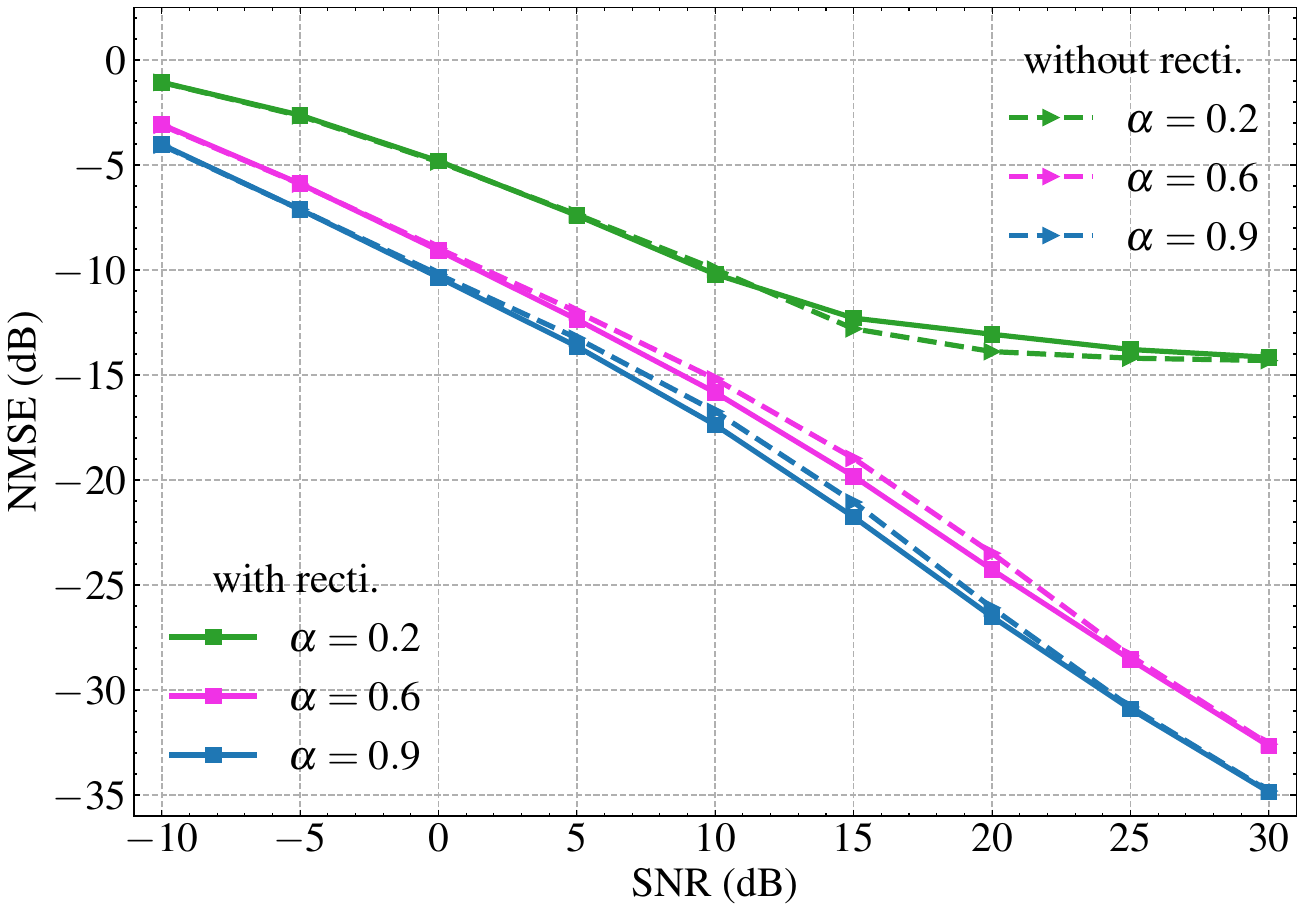}}  \vspace{-0mm}
		\captionsetup{font=footnotesize, name={Fig.}, labelsep=period} 
		\caption{\, Effect of anchored rectification on the CDL-C channels.}
		\label{fig:anchored_rectification} \vspace{-8mm}
	\end{center}
\end{figure}

\begin{figure}[t]\vspace{-2mm}
	\begin{center}
		\centerline{\includegraphics[width=0.45\textwidth]{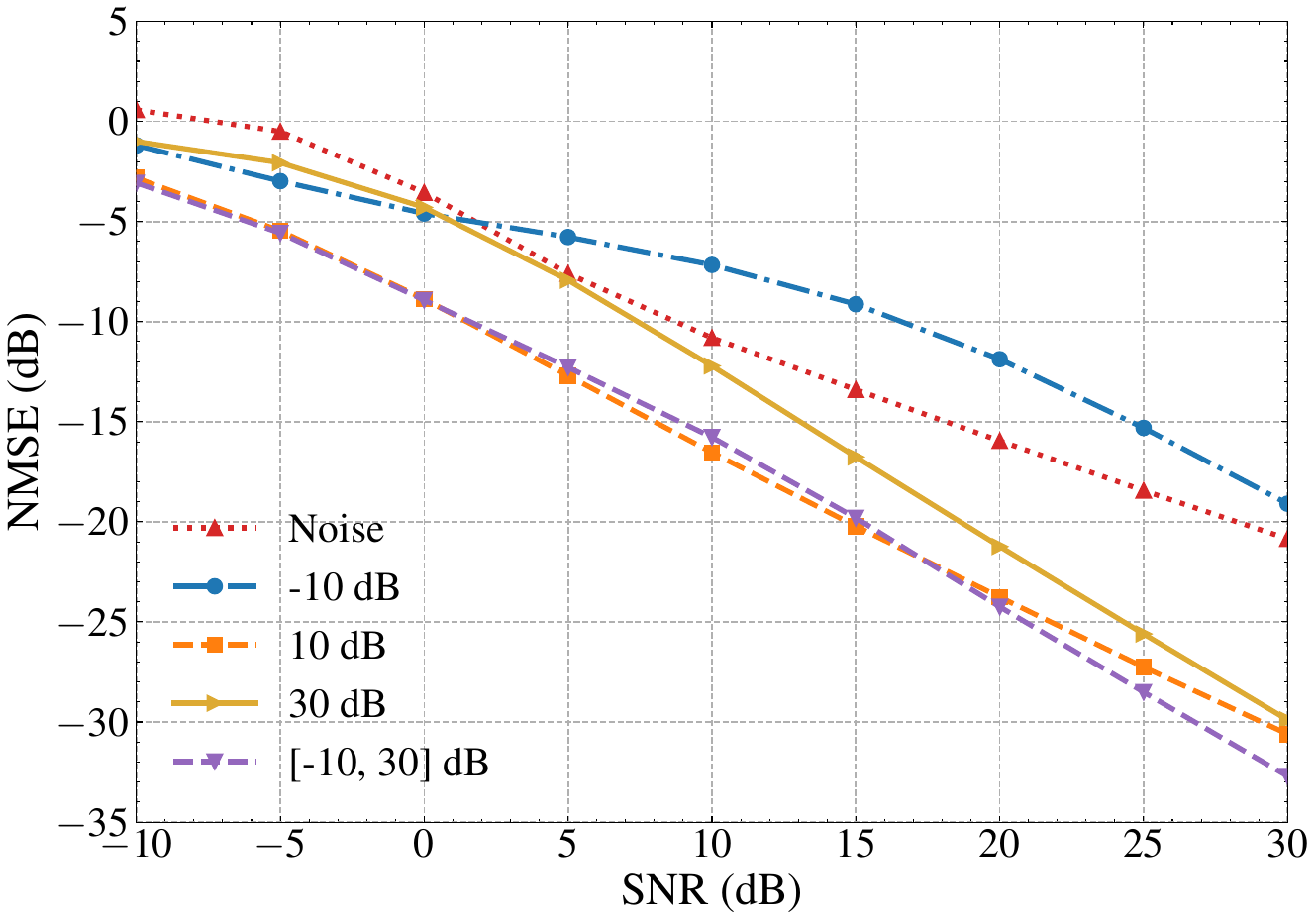}}  \vspace{-0mm}
		\captionsetup{font=footnotesize, name={Fig.}, labelsep=period} 
		\caption{\, Ablation on the training SNR. The models are trained on clean CDL-C channels with different synthetic SNR levels, where ``noise" denotes pure Gaussian noise.}
		\label{fig:perturb_snr_ablation} \vspace{-8mm}
	\end{center}
\end{figure}

\begin{figure}[t]\vspace{-2mm}
	\begin{center}
		\centerline{\includegraphics[width=0.45\textwidth]{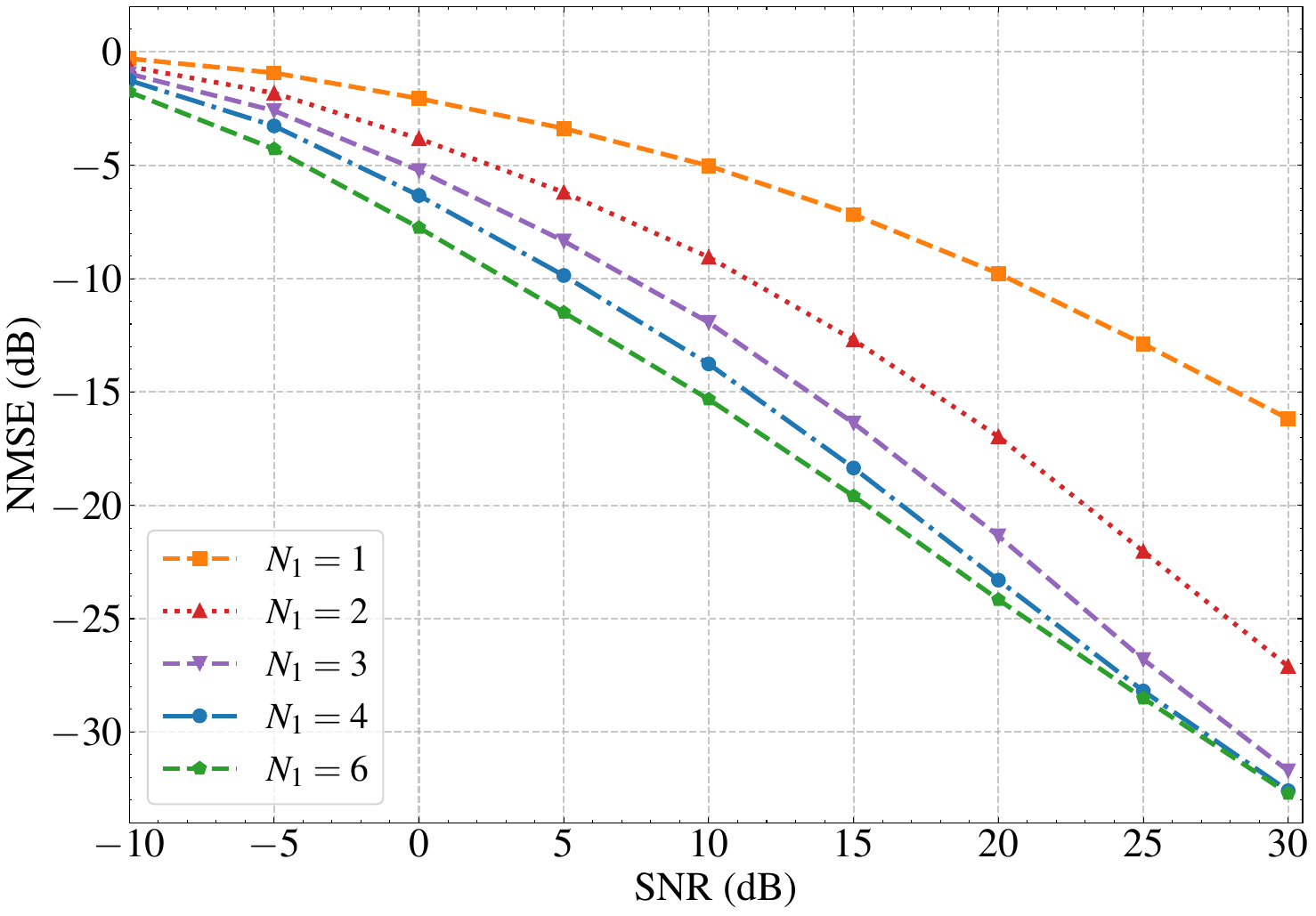}}  \vspace{-0mm}
		\captionsetup{font=footnotesize, name={Fig.}, labelsep=period} 
		\caption{\, Ablation study of the recursive mechanism. The evaluation is conducted with $\alpha =0.6$  trained and tested on CDL-C.}
		\label{fig:Ablation_N1_SNR} \vspace{-8mm}
	\end{center}
\end{figure}

\section{Ablation Studies}
\setcounter{figure}{0}
We provide additional ablation studies to further validate the robustness and effectiveness of the proposed method. 
\subsection{Impact of Anchored Rectification}
Fig. \ref{fig:anchored_rectification} evaluates the NMSE performance with and without anchored rectification across various pilot densities. In moderate-to-high pilot regimes ($\alpha \geq 0.6$), the anchored mechanism provides a modest but consistent accuracy gain, typically ranging from $0.3$ dB to $0.8$ dB. These results indicate that rectification helps stabilize the iterative refinement and suppresses potential trajectory drift as the measurement guidance becomes more reliable. Conversely, in the highly undersampled regime ($\alpha=0.2$), the performance gain becomes marginal since the reconstruction is primarily limited by measurement ambiguity rather than trajectory stability. This confirms that anchored rectification is most effective in practical pilot configurations, ensuring a robust equilibrium between the generative prior and physical observations.

\subsection{Impact of Training SNR Strategy}
Fig. \ref{fig:perturb_snr_ablation} shows that the mixed training SNR strategy ($[-10, 30]$ dB) achieves the best overall performance. In contrast, both the noise baseline and the $-10$ dB setting perform poorly, while the $30$ dB model shows mediocre results at low testing SNRs. Notably, the $10$ dB setting is highly effective and closely approaches the mixed-range performance. Ultimately, the mixed perturbation strategy is superior as it enhances the model's generalization across the diverse noise levels encountered during the recursive refinement process.

\subsection{Impact of the Recursive Mechanism}
Fig. \ref{fig:Ablation_N1_SNR} evaluates the contribution of the recursive mechanism. The $N_1=1$ curve represents a non-recursive inference without anchor refinement. It is evident that the non-recursive version performs significantly worse than the recursive configurations across the entire SNR spectrum. Notably, at a high SNR of 30 dB, the RC-Flow ($N_1=6$) achieves a performance gain of approximately 16 dB over the non-recursive version ($N_1=1$). This substantial gap confirms that a single traversal of the flow trajectory is insufficient to overcome discretization errors and trajectory drift, emphasizing that recursive anchor refinement is critical for steering the estimate toward the true data manifold.

\begin{figure*}[t]\vspace{-2mm}
	\begin{center}
		\centerline{\includegraphics[width=0.95\textwidth]{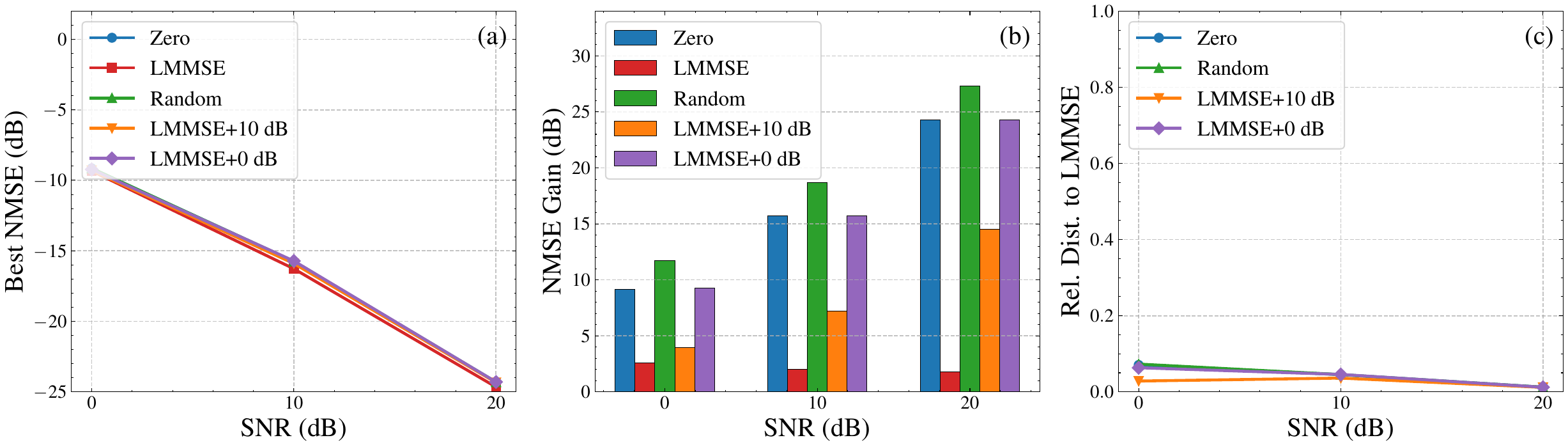}}  \vspace{-0mm}
		\captionsetup{font=footnotesize, name={Fig.}, labelsep=period} 
		\caption{\, Multi-initialization stability on CDL-C. 
(a) Best NMSE achieved during the recursion under different initialization strategies. 
(b) NMSE gain from the initial estimate to the best iteration. 
(c) Relative distance between the best reconstructed state and the LMMSE-initialized solution, indicating whether different initializations lead to distinct solutions.}
		\label{fig:multi_init} \vspace{-8mm}
	\end{center}
\end{figure*}

\begin{figure*}[t]\vspace{-2mm}
	\begin{center}
		\centerline{\includegraphics[width=0.95\textwidth]{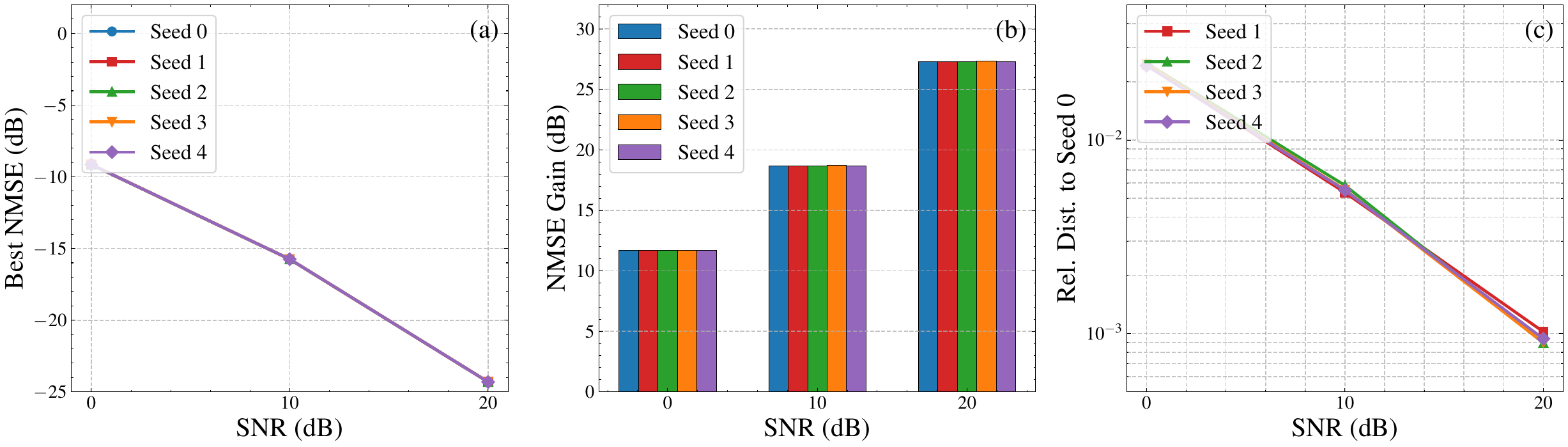}}  \vspace{-0mm}
		\captionsetup{font=footnotesize, name={Fig.}, labelsep=period} 
		\caption{\, Random-initialization seed sensitivity on CDL-C. 
(a) Best NMSE obtained with different random seeds. 
(b) NMSE gain from each random initialization to its best iterate. 
(c) Relative distance between the best reconstructed state of each seed and that of Seed~0, measuring seed-induced solution variability.}
		\label{fig:random_seed_init} \vspace{-8mm}
	\end{center}
\end{figure*}

\subsection{Fixed-point Convergence under Different Initializations.}
Theoretically, the solver is expected to converge into a stable neighborhood of the posterior distribution rather than a unique singular point. To validate this empirically, we conduct multi-initialization studies (Fig.~\ref{fig:multi_init}) and random seed sensitivity analyses (Fig.~\ref{fig:random_seed_init}) on CDL-C channels.
\par
Fig.~\ref{fig:multi_init} evaluates the recursion under diverse starting points, including zero, LMMSE, random, and perturbed LMMSE estimates. The results demonstrate that the solver consistently refines all initial estimates, regardless of whether they are weak or strong, to achieve nearly identical reconstruction accuracy. The relative distance between different final states remains extremely small, and the final fixed-point residuals decrease from $1.6\times10^{-3}$ at $0$ dB to about $2.0\times10^{-5}$ at $20$ dB, confirming that the iteration is highly insensitive to the choice of initialization.
\par

Fig.~\ref{fig:random_seed_init} analyzes the impact of different random initialization seeds. Although different seeds may settle on slightly distinct points within the stable neighborhood, their relative distance consistently remains below $10^{-2}$, yielding functionally equivalent NMSE performance and gains. Across all evaluations, we observed no instances where initializations led to distinct undesired or divergent solutions, confirming the robustness and empirical stability of the proposed framework.

\end{appendices}

\small
\bibliographystyle{IEEEtran}
\bibliography{bib}
\vspace{12pt}

\end{document}